\documentclass[letterpaper,journal]{IEEEtran}
\IEEEoverridecommandlockouts

\usepackage{cite}
\usepackage{amssymb,amsfonts}
\usepackage{graphicx}
\usepackage{textcomp}
\usepackage{xcolor}
\usepackage[english]{babel}
\usepackage[compatibility=false]{caption}
\usepackage{subcaption}
\usepackage{balance}
\usepackage{times}
\usepackage{url}
\usepackage[hidelinks]{hyperref}
\usepackage[utf8]{inputenc}

\usepackage{amsthm}
\usepackage{booktabs}
\usepackage[ruled,vlined,linesnumbered]{algorithm2e}
\usepackage{fix-cm}

\usepackage{tabularx}
\usepackage{amsmath}
\newtheorem{theorem}{\textbf{Theorem}}

\newcolumntype{Y}{>{\centering\arraybackslash}X}
\definecolor{commentcolor}{RGB}{0, 0, 0} 

\begin{document}

\title{From Cloud to Crowd: Democratizing LLM Service with Decentralized Edge Collaboration for RAG
}

\author{Jiaxing Li, Hengzhi Wang, Feng Wang,~\IEEEmembership{Senior Member,~IEEE}, Chi Xu, Danyang Song, Ruixiao Zhang,\\ Edith C. H. Ngai,~\emph{Senior Member, IEEE}, Jiangchuan Liu,~\IEEEmembership{Fellow,~IEEE}

\thanks{\copyright~2026 IEEE. Personal use of this material is permitted. Permission from IEEE must be obtained for all other uses, in any current or future media, including reprinting/republishing this material for advertising or promotional purposes, creating new collective works, for resale or redistribution to servers or lists, or reuse of any copyrighted component of this work in other works. Accepted for publication in IEEE Transactions on Mobile Computing.}
\thanks{This work is partly supported by an NSERC Discovery Grant. (Corresponding author: Jiangchuan Liu.)}
\thanks{Jiaxing Li, Chi Xu and Jiangchuan Liu are with the School of Computing Science, Simon Fraser University, Burnaby, BC V5A 1S6, Canada (e-mail: \{jla641, chix, jcliu\}@sfu.ca).}
\thanks{Hengzhi Wang is with the College of Computer Science and Software Engineering, Shenzhen University, Shenzhen 518060, China (e-mail: whz@szu.edu.cn).}
\thanks{Feng Wang is with the Department of Computer and Information Science, The University of Mississippi, University, MS 38677, USA (e-mail: fwang@cs.olemiss.edu).}
\thanks{Danyang Song and Edith C. H. Ngai are with the Department of Electrical and Electronic Engineering, University of Hong Kong, Pok Fu Lam, Hong Kong (e-mail: dysong@connect.hku.hk, chngai@eee.hku.hk).}
\thanks{Ruixiao Zhang is with the Department of Computer Science, The University of Illinois Urbana-Champaign, Champaign, IL 61820, USA (e-mail: ruixiao.cs.zhang@gmail.com).}
}
\maketitle

\begin{abstract}
The rapid advancement of large language models (LLMs) has increased demand for scalable and cost-effective deployment, especially for mobile and edge devices. Cloud-hosted LLMs are powerful but expensive and difficult to scale due to vendor lock-in and high resource needs, resulting in high expenses and unstable performance under load. Recent efforts focus on deploying small language models (SLMs), distilled or pruned from LLMs, on resource-constrained edge devices to reduce costs and improve scalability. However, edge-based SLMs face limited knowledge coverage and notable accuracy gap compared to cloud-based LLMs. To address this, we present DEFRAG, a decentralized edge collaboration system for retrieval-augmented generation (RAG) that optimizes both retrieval and generation across heterogeneous edge devices. For retrieval, DEFRAG compresses and shares knowledge graphs, using hybrid retrieval to expand knowledge coverage. For generation, DEFRAG introduces an optimizer that adaptively selects SLMs and RAG parameters per query, balancing accuracy and cost. We implement DEFRAG on a heterogeneous edge testbed and evaluate it on benchmark QA datasets. We also test it under mobile route stress, non-uniform data placement, and a domain-specific QA workload. The results show that DEFRAG maintains stable service quality and cost efficiency under these broader settings. Results show that DEFRAG narrows the SLM–LLM accuracy gap, while reducing cost by up to 98.4\% and increasing peak throughput by up to 97.8\% over centralized services. These findings demonstrate the potential of DEFRAG for democratized LLM services at the edge.
\end{abstract}


\section{Introduction}
\IEEEPARstart{T}{he} remarkable progress in large language models (LLMs) has accelerated the proliferation of AI-powered applications across a range of domains, including content creation, intelligent assistants, and collaborative robotics~\cite{qu2025mobile}. Increasingly, these applications are accessed from mobile and edge clients (e.g., smartphones, laptops, wearables, and IoT gateways), where usage is often frequent and interactive. As a result, per-query and per-token costs become a first-order barrier to scalability and accessibility. However, the prevailing paradigm for deploying LLMs remains largely cloud-centric, relying on centralized supercomputing resources to deliver state-of-the-art performance~\cite{wang2023tabi}.

While this approach enables access to powerful models, it is constrained by two critical limitations. These limitations are particularly costly for mobile and edge clients, where frequent interactive queries make per-token charges add up quickly. First, the vendor-locked nature of cloud-based LLM services raises the barriers to entry for widespread adoption. Subscription and token-based pricing schemes make high-quality language models prohibitively expensive for many users, hindering broader accessibility~\cite{chenfrugalgpt,zhang2024edgeshard}. Furthermore, monolithic model architectures complicate the rapid integration of newly emerging knowledge, rendering real-time personalization and timely adaptation challenging and inefficient~\cite{asai2024self,gao2023retrieval}. Second, centralized frameworks suffer from limited scalability. As the demand for LLM services grows, vertical scaling—adding capacity within individual data centers—leads to steeply increasing costs across hardware, software, and energy consumption~\cite{samsi2023words}. These costs are further exacerbated by inherent resource utilization imbalances. During off-peak periods, resources are underutilized; during peak demand, oversubscription leads to both compute and network congestion, causing unpredictable tail latency and limiting scalability~\cite{zhao2023scalable,reidys2025coach}.

\begin{figure}[t]
\centering
\includegraphics[width=\columnwidth]{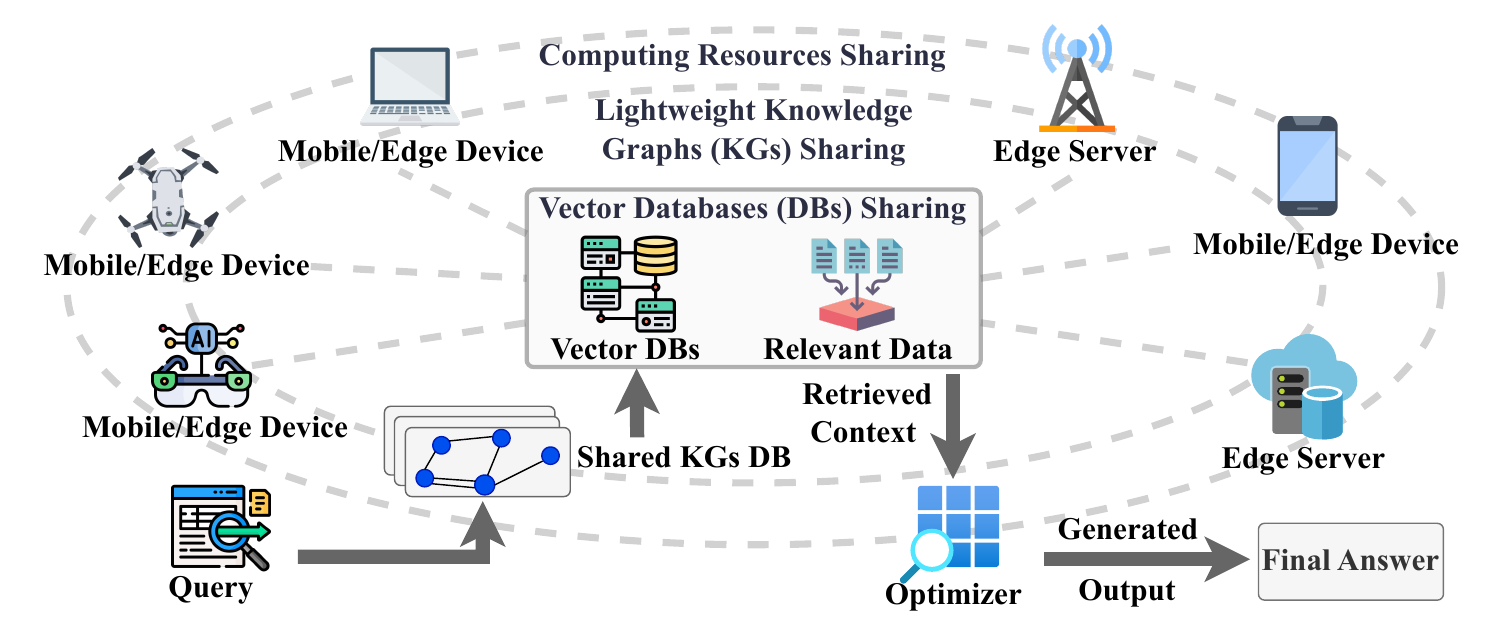}
\caption{Overview of the \textbf{D}ecentralized \textbf{E}dge Collaboration \textbf{F}or \textbf{RAG} (DEFRAG) system architecture.}
\label{fig:overview}
\end{figure}

To address the high cost and limited scalability of centralized LLM services, leveraging idle computing resources on edge devices for distributed deployment has emerged as a promising approach~\cite{huang2024toward}. Existing research mainly follows two directions. The first partitions large-parameter LLMs across edge nodes, assigning inference tasks based on bandwidth and processing capacity to reduce the load on centralized systems~\cite{he2024large,zhang2024edgeshard}. However, this method introduces significant coordination and communication overhead, resulting in increased latency and possible inference stalls~\cite{zhang2024edgeshard}. The second direction deploys language models of varying sizes, especially small language models (SLMs), independently on edge devices. SLMs are typically obtained by distilling or pruning larger LLMs~\cite{qwen3}, making them lightweight variants of LLMs with much lower resource requirements~\cite{wang2023tabi}. Integrating retrieval-augmented generation (RAG) at the edge further expands SLM knowledge coverage by retrieving external database content, narrowing the knowledge gap between SLMs and LLMs, and improving answer accuracy~\cite{shojaee2025federated,ouyang2025adarag}.

Our measurement shows that SLMs on individual edge devices are still constrained by resource capacity (Section~\ref{section:measurement}). Most mobile edge devices can support SLMs up to 4B parameters; without further optimization, the accuracy gap with a 235B LLM can reach 38\%. Access to external databases via RAG can reduce this gap to 16.5\%, enabling edge-based SLMs to approach the accuracy of cloud-based LLMs. Nevertheless, single-device resource limits remain a bottleneck. Edge collaboration is therefore essential to narrow the accuracy gap and enhance cost efficiency and scalability.

In this context, we propose the \textbf{D}ecentralized \textbf{E}dge Collaboration \textbf{F}or \textbf{RAG} (DEFRAG) system architecture, which organizes heterogeneous edge nodes into a unified network that shares both data storage and computational resources. To make LLM-quality service affordable and scalable for mobile and edge clients, decentralized RAG must address several challenges. First, device heterogeneity complicates system coordination. Edge devices vary greatly in computational capacity, supporting different model sizes and quantization precisions, making unified SLM inference challenging. Second, query-aware adaptation is challenging. Aligning RAG parameters and SLM selection with each query’s demands is nontrivial, as both excessive and insufficient retrieval can degrade performance, and SLMs differ in their ability to process retrieved content. Third, network-wide knowledge fragmentation limits coverage. Comprehensive knowledge bases cannot be stored on individual edge devices; thus, efficient knowledge sharing and collaboration across the network are essential for maintaining inference quality, yet this remains an open challenge.

To address these challenges, we design the DEFRAG system architecture as shown in Figure~\ref{fig:overview}. First, to overcome device storage limits and knowledge fragmentation, we compress knowledge graphs into lightweight entity-relation representations without embeddings, enabling efficient caching and sharing across edge devices. For data in vector databases, sparse retrieval over these lightweight graphs identifies relevant entities and fetches the corresponding text and embeddings from peer devices, supporting broad and deep knowledge through cross-device collaboration. Second, motivated by our data analysis, we introduce a query-adaptive optimizer based on Bayesian online learning. Unlike linear-contextual bandits, Bayesian optimization better models complex, nonlinear interactions among queries, retrieval contexts, and generation configurations, providing robust accuracy even with limited data. The optimizer updates predictions through online feedback, dynamically selecting RAG parameters and SLM configurations to balance accuracy and cost efficiency. Finally, to handle resource heterogeneity, devices are grouped by computational capacity, and the optimizer assigns models of appropriate size and precision to peer devices, enabling coordinated and efficient SLM inference. These advances allow SLMs to approach the accuracy of centralized LLMs while improving cost efficiency and scalability, supporting broader and more affordable access to LLMs.

We implemented and evaluated DEFRAG on a testbed with heterogeneous edge devices, analyzing its performance across accuracy, cost, and scalability. For accuracy, the optimizer’s adaptive selection of RAG parameters and SLM configurations narrowed the gap between edge-based SLMs and cloud-based LLMs. On two benchmark QA datasets, edge accuracy improved from 61.5\% and 76.6\% to 67.8\% and 81.4\%, reducing the accuracy gap from 12.7\% and 7.1\% to 6.4\% and 2.3\%. For cost, DEFRAG reduced expenses by up to 98.4\% with full peer resource sharing. For scalability, DEFRAG achieved up to 97.8\% higher peak throughput than centralized services as concurrent user numbers increased. We also test DEFRAG under mobile route stress, non-uniform data placement, and a domain-specific QA workload. The results show that DEFRAG maintains stable service quality and cost efficiency under these broader settings. These results show that DEFRAG effectively leverages heterogeneous edge resources to deliver accurate, cost-efficient, and scalable LLM services.

DEFRAG connects edge devices in a decentralized network, enabling users to share and access models and knowledge. This approach shifts LLM service from centralized cloud delivery to collaborative provisioning by the crowd, advancing the democratization of LLM capabilities. 

Our contributions can be summarized as follows:
\begin{itemize}
    \item We propose a lightweight edge-collaborative RAG retrieval strategy that compresses and distributes knowledge graphs for efficient coverage under storage constraints.
    \item We design an optimizer that adapts RAG generation parameters per query and updates online from historical data, balancing accuracy and cost efficiency across heterogeneous edge devices.
    \item We implement and deploy the DEFRAG prototype on a heterogeneous edge devices testbed. Extensive experiments demonstrate that DEFRAG significantly narrows the SLM–LLM accuracy gap, achieves substantial cost savings, and provides strong scalability compared to centralized baselines.
\end{itemize}

\section{Motivation and Observation}
\label{section:measurement}
In this section, we comprehensively evaluate recent edge SLMs and their gap with LLM on benchmark QA datasets, analyzing performance across model scales, quantization, query types, devices, and RAG strategies. The results reveal key challenges and opportunities for improvement.

\subsection{Preliminary}
\label{preliminary}

We evaluate both SLM and LLM models from the Qwen3 series~\cite{qwen3}, as it provides high-quality, fully open-source models at multiple parameter sizes for fair comparison. The SLMs include 1.7B, 4B, 8B, and 14B variants, with Qwen3-32B as an edge-server reference point and Qwen3-235B as the reference LLM. For comprehensive comparison, we consider 4-bit, 8-bit, and half-precision quantizations, implemented via llama.cpp~\cite{llama.cpp}. Table~\ref{tab:qwen3_model_sizes} reports GPU memory requirements for each configuration.

To assess practical performance on heterogeneous edge devices, we evaluate five representative edge devices: Jetson Orin Nano, Jetson AGX Orin, Galaxy S25 Ultra (Snapdragon 8 Elite), RTX 4090, and RTX 5090. These devices cover both PC-class and mobile platforms for applications such as smartphones, laptops, AR glasses, and robots.
For question-answer (QA) evaluation, we use Natural Questions (NQ)\cite{kwiatkowski2019natural} and HotpotQA\cite{yang2018hotpotqa}, which contain single-hop and multi-hop queries, respectively. To ensure consistency in RAG and GraphRAG evaluation,  we sample top 500 QA pairs from each dataset by selecting questions supported by Wikipedia pages referenced more than 10 times.

By leveraging RAG, SLMs can access the same breadth of knowledge as LLMs, thereby narrowing the accuracy gap caused by knowledge limitations. To measure the current gap between SLMs and LLMs, we adopt two advanced RAG methods: standard dense retrieval~\cite{karpukhin2020dense} and a knowledge graph-based approach adapted from GraphRAG~\cite{edge2024local}.

\begin{table}[t]
\centering
\footnotesize
\renewcommand{\arraystretch}{1.0}
\caption{Qwen3 models memory usage across different parameter size and quantization precisions.}
\begin{tabularx}{\linewidth}{lXXX}
    \toprule
    \textbf{Models} & \textbf{INT4} & \textbf{INT8} & \textbf{FP16} \\
    \midrule
    Qwen3-1.7B & 1.4 GB & 2.2 GB & 4.1 GB \\
    Qwen3-4B & 2.6 GB & 4.4 GB & 8.1 GB \\
    Qwen3-8B & 5.2 GB & 8.9 GB & 16 GB \\
    Qwen3-14B & 9.3 GB & 16 GB & 30 GB \\
    Qwen3-32B & 20 GB & 35 GB & 66 GB \\
    Qwen3-235B & 142 GB & 250 GB & 470 GB \\
    \bottomrule
\end{tabularx}
\label{tab:qwen3_model_sizes}
\end{table}

\subsection{Data-Driven Analysis}
\label{subsection:Data-Driven Analysis}

\begin{figure}[t]
    \centering
    \includegraphics[width=0.8\linewidth]{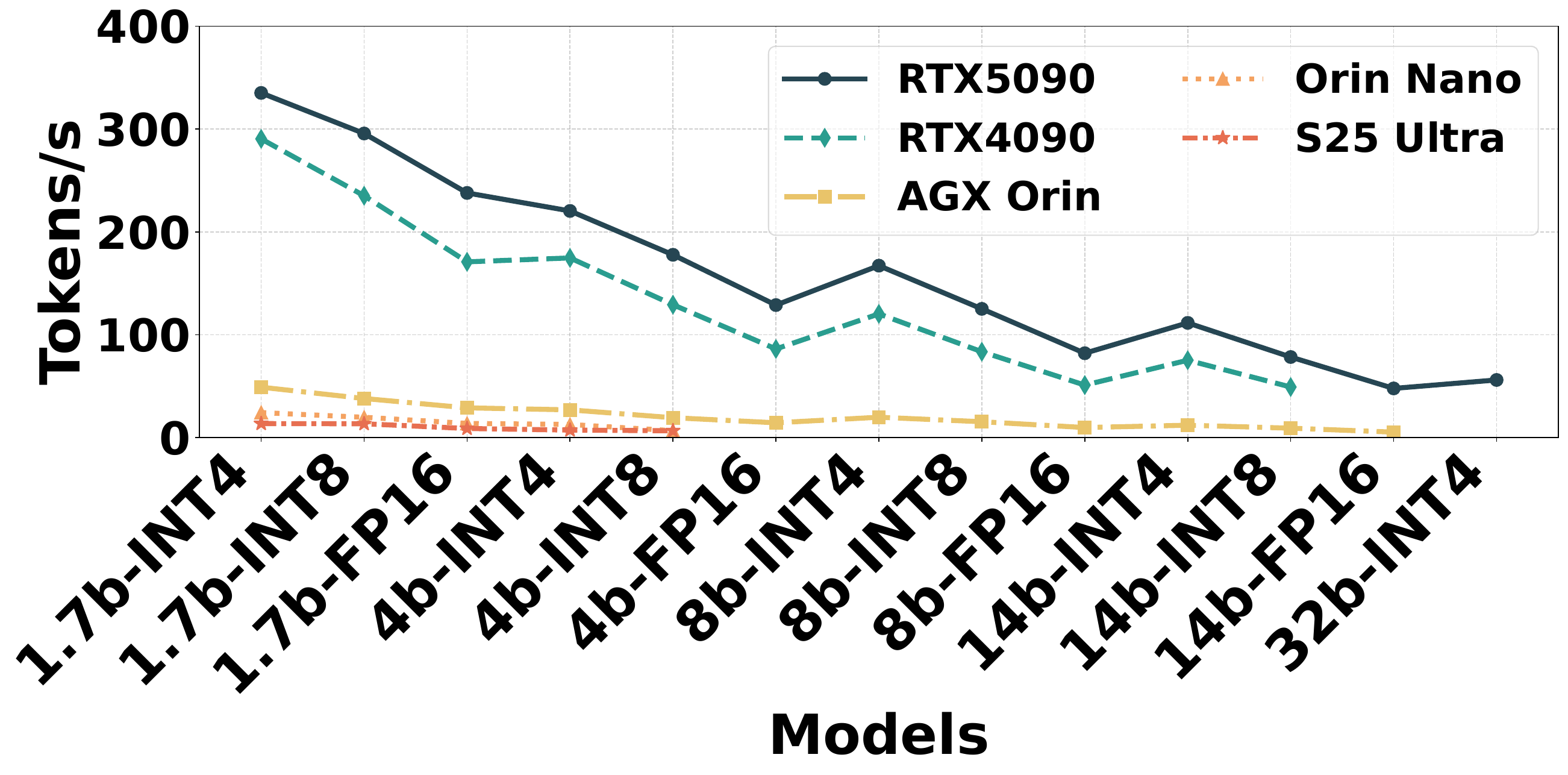}
    \caption{Average throughput (tokens/s) of various SLMs across heterogeneous edge devices.}
    \label{fig:llm_tokens}
\end{figure}

Edge computing environments comprise a diverse array of devices, from high-performance servers to lightweight mobile hardware, with significant variability in practical performance. Prior research has seldom provided fine-grained evaluation of the capabilities of these devices, particularly for SLMs of different sizes and quantization precisions. To address this, we systematically measure the performance of representative edge devices. As shown in Figure~\ref{fig:llm_tokens}, these devices can be grouped by their token generation rates: edge servers (such as RTX 4090 and RTX 5090) can stably serve larger models (e.g., 14B), while edge devices (such as Jetson Orin Nano, Jetson AGX Orin, and Galaxy S25 Ultra) typically support smaller models (e.g., 1.7B). Edge servers provide stable, consistent performance, whereas edge devices display greater variability. Effectively coordinating such heterogeneous devices, each with distinct capabilities, to collaboratively serve SLMs as a unified system remains an open challenge.

Recent work has explored combining RAG with SLMs on edge devices~\cite{ouyang2025adarag,shojaee2025federated,fan2025minirag}, but systematic evaluations across model sizes, quantization precisions, RAG strategies, and query types remain limited—particularly regarding the accuracy gap between SLMs and LLMs. To address this, we use the Qwen3 series to evaluate the SLM--LLM accuracy gap on Natural Questions and HotpotQA using standard RAG and GraphRAG. Accuracy is assessed by GPT-4.1~\cite{openai2025gpt41api} on 500 QA pairs, using a binary scoring method~\cite{wangpandalm}.

As shown in Figure~\ref{fig:RAG_comparison}, for Natural Questions, standard RAG with a static TOP-K=7 strikes a good balance between context sufficiency and overload, significantly narrowing the accuracy gap between SLMs and LLMs. In contrast, on the multi-hop HotpotQA dataset, standard RAG fails to capture logical relationships, resulting in limited accuracy improvements. We therefore adopt a simplified version of GraphRAG, since the dense retrieval results of the GraphRAG approach often exceed the processing capacity of SLMs~\cite{fan2025minirag}. Specifically, we retain only essential information—the title and description of the top 10 entities, the source, target, and description of the top 10 relationships, and the text of the top 2 data chunks, each selected using static heuristics. This streamlined retrieval is better suited to SLM capacity, improving reasoning accuracy and further reducing the SLM–LLM gap on HotpotQA.

\begin{figure}[t]
    \centering
    \includegraphics[width=0.8\linewidth]{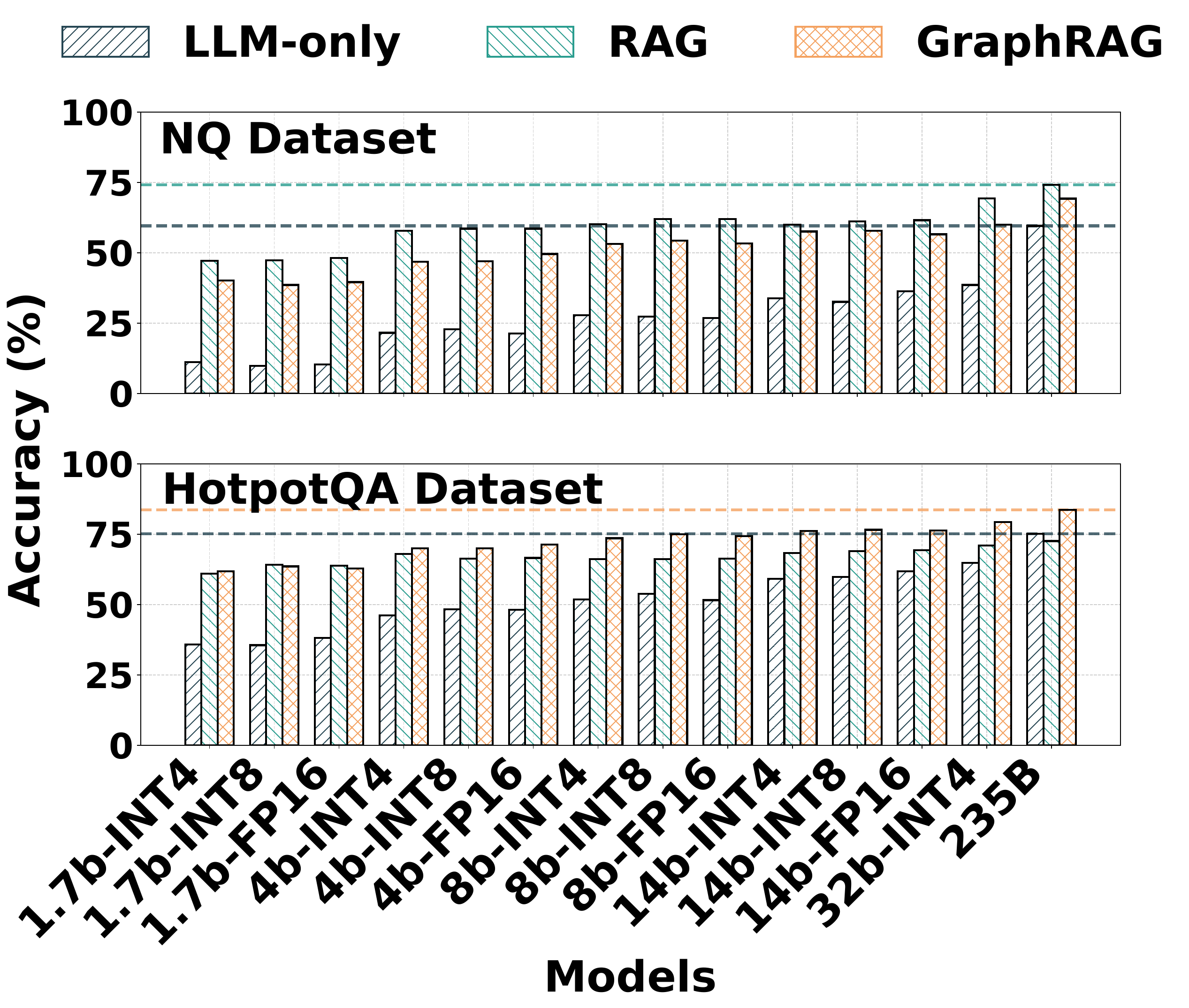}
    \caption{Accuracy comparison on NQ and HotpotQA with direct generation, RAG, and GraphRAG.}
    \label{fig:RAG_comparison}
\end{figure}

These results indicate that RAG can substantially narrow the accuracy gap between SLMs and LLMs. However, the optimal RAG configuration, including the number of entities, relationships, and data chunks, varies by query type (e.g. single-hop vs. multi-hop). Standard RAG can be considered a special case of GraphRAG where no entities and relationships are retrieved. Achievable accuracy also depends on the SLM chosen. Therefore, dynamically selecting both SLMs and RAG parameters for each query remains a key challenge.

As shown in Table~\ref{tab:threshold_coverage}, raising the relevance score\footnote{Indicates how well retrieved content relates to the query.} threshold for retrieved content allows SLMs to maintain higher accuracy, though with reduced coverage\footnote{Proportion of queries with average relevance score above threshold.}. This suggests that leveraging relevance scores to characterize query patterns and guide RAG generation parameters tuning can further close the gap to LLM accuracy on suitable queries. Adaptive selection of SLMs and RAG configurations based on query relevance score is therefore a promising approach for accurate edge inference.

\begin{table}[t]
\centering
\footnotesize
\renewcommand{\arraystretch}{1.0}
\caption{Coverage of QA pairs and mean SLMs accuracy at different entity relevance score thresholds (Dataset: HotpotQA, Method: GraphRAG)}
\begin{tabularx}{\linewidth}{lYYY}
    \toprule
    \textbf{Threshold} & $\theta \ge 0.50$ & $\theta \ge 0.55$ & $\theta \ge 0.65$ \\
    \midrule
    \textbf{Coverage}  & 459/500 (91.8\%) & 363/500 (72.6\%) & 77/500 (15.4\%) \\
    \midrule
    \textbf{Accuracy} (1.7B)  & 65.0\% & 67.6\% & 80.1\% \\
    \textbf{Accuracy} (4B)  & 72.7\% & 75.2\% & 83.5\% \\
    \textbf{Accuracy} (8B)  & 76.4\% & 78.8\% & 83.5\% \\
    \textbf{Accuracy} (14B) & 78.2\% & 80.9\% & 88.7\% \\
    \bottomrule
\end{tabularx}
\label{tab:threshold_coverage}
\end{table}

Another observation is that a complete graph database, including both the knowledge graph and vector database, can be large. Even when limited to popular Wikipedia pages from NQ and HotpotQA, database sizes reach 983.83MB (139 pages) and 1473.98MB (231 pages), respectively. In practice, each edge device stores only a subset, limiting knowledge coverage. Queries involving out-of-scope content must be broadcast to neighboring devices, which is inefficient and difficult to handle at scale. Efficient data sharing and collaboration among edge devices therefore remains a key challenge.

We observe that the knowledge graph database without embeddings (entities and relationships only) is much smaller—about 1/35 the size of the vector database (see Table~\ref{tab:db_sizes_vertical}). This suggests an effective strategy: cache shared, lightweight knowledge graphs on each edge device, and use sparse retrieval~\cite{karpukhin2020dense} to identify relevant content in neighboring vector databases, enabling accurate extraction of data chunks and embeddings. This approach efficiently expands both the breadth and depth of knowledge accessible to each device.

\begin{table}[t]
\centering
\footnotesize
\renewcommand{\arraystretch}{1.0}
\caption{Storage sizes of the Knowledge Graph database (KG DB) and Vector database (Vector DB) for the NQ and HotpotQA datasets.}
\begin{tabularx}{\linewidth}{@{}lYYY@{}}
    \toprule
    \textbf{Dataset} & \textbf{KG DB} & \textbf{Vector DB} & \textbf{Scale} \\
    \midrule
    Natural Questions~\cite{kwiatkowski2019natural} & 27.42MB & 956.41MB & $\sim$35$\times$ \\
    HotpotQA~\cite{yang2018hotpotqa} & 40.38MB & 1433.6MB & $\sim$36$\times$ \\
    \bottomrule
\end{tabularx}
\label{tab:db_sizes_vertical}
\end{table}

\subsection{Challenges and Opportunities}
Based on our data-driven analysis, several key challenges and corresponding opportunities emerge in advancing distributed LLM services with the DEFRAG system architecture, as summarized below:

\textbf{Challenge 1: Coordinated Generation on Heterogeneous Edge Devices.} Edge environments consist of diverse device types with significant variability in computational capacity, supporting different model sizes and quantization precisions (Figure~\ref{fig:llm_tokens}). Effectively coordinating these heterogeneous devices to serve SLMs as a unified system remains a challenge.

\textbf{Opportunity:} Systematic grouping of devices based on capability, together with dynamic resource scheduling that assigns models of suitable size and precision to corresponding hardware, can improve overall resource efficiency. This approach enables coordinated and effective SLM inference, facilitating integrated utilization of edge resources.

\textbf{Challenge 2: Query-Adaptive Selection of RAG Generation Parameters.} As shown in Figure~\ref{fig:RAG_comparison}, aligning RAG retrieval parameters with SLM capacity is crucial for inference accuracy and efficiency. Both excessive and insufficient information can degrade performance, and SLMs differ in their ability to process retrieved content. Dynamically selecting parameters for each query, therefore, remains a key challenge.

\textbf{Opportunity:} As shown in Table~\ref{tab:threshold_coverage}, analyzing relevance scores between the query and retrieved context enables dynamic selection of SLM configurations and RAG parameters for each query. This approach allows edge devices to achieve high accuracy while maintaining efficient resource utilization.

\textbf{Challenge 3: Scalable Knowledge Retrieval for RAG Across Edge Devices.}
Comprehensive knowledge bases, including knowledge graphs and vector databases, are typically too large for individual edge devices to store in full, resulting in limited knowledge coverage per device. Efficient cross-device knowledge sharing and processing remain a challenge.

\textbf{Opportunity:} Compressing the knowledge graph into a lightweight representation of essential entities and relationships enables efficient caching across edge devices, ensuring broad knowledge coverage (Table~\ref{tab:db_sizes_vertical}). Sparse retrieval can then dynamically locate relevant entities and access corresponding vector data from peer devices to address depth requirements. This cross-device collaboration efficiently expands knowledge scope and improves overall coverage.

\section{System Design and Modeling}
\begin{figure}[t]
\centering
\includegraphics[width=\columnwidth]{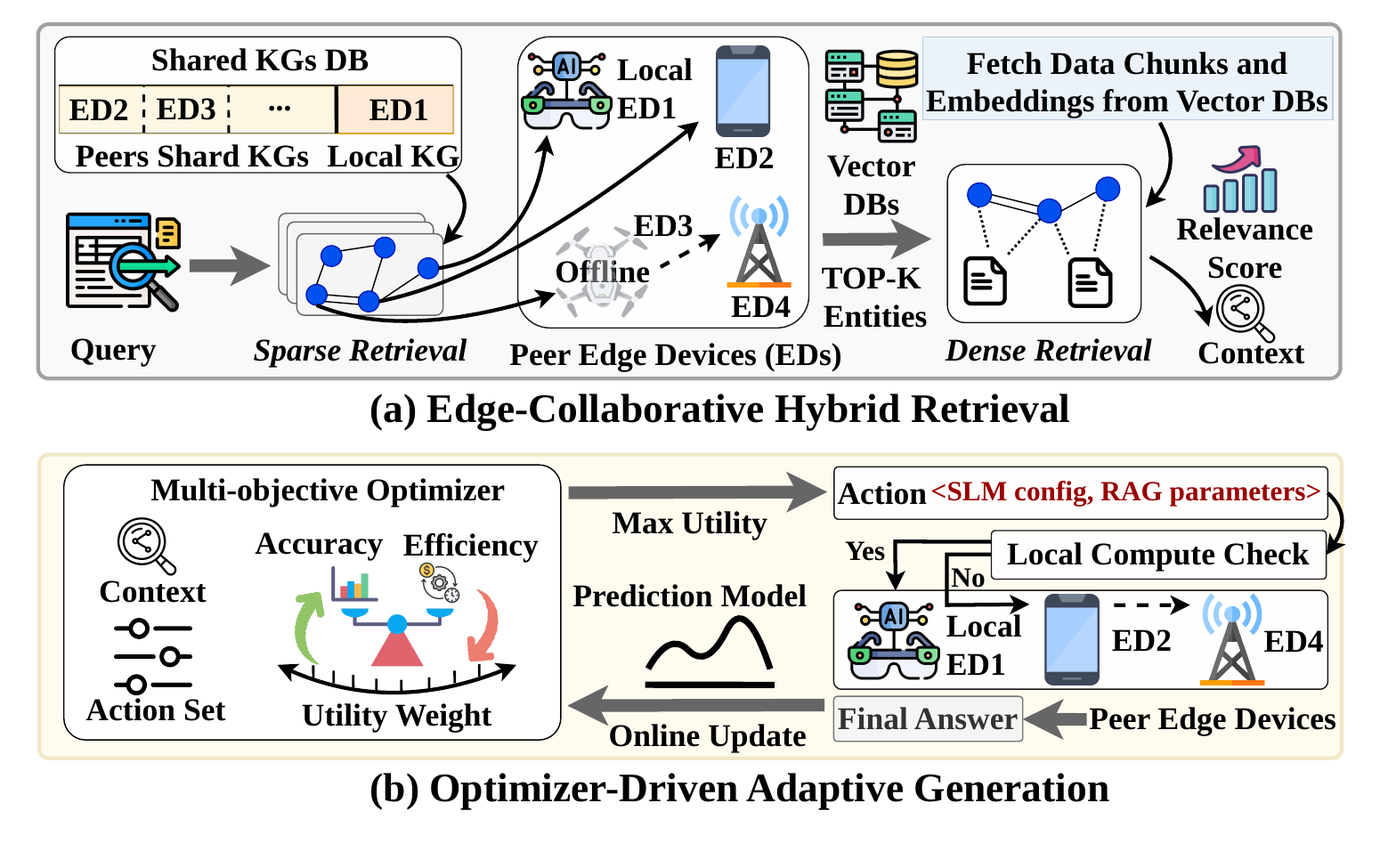}
\caption{Design of the DEFRAG system architecture and end-to-end workflow, combining edge-collaborative hybrid retrieval with optimizer-driven adaptive generation and optional peer offloading.}
\label{fig:framework_design}
\end{figure}

In this section, we present DEFRAG, a decentralized RAG architecture that enables collaboration among edge devices to address the challenges above. We first provide an overview of the architecture, followed by details on the modeling of its core components.

\subsection{Methodology Overview}

DEFRAG enhances both retrieval and generation across a group of edge devices. For retrieval, it expands knowledge access by sharing lightweight knowledge graphs among peers. For generation, it adaptively tunes SLMs and RAG parameters based on query-specific context features extracted during retrieval. This design balances accuracy and efficiency, enabling the system to approach cloud LLM accuracy while running SLMs on heterogeneous edge devices. Figure~\ref{fig:framework_design} shows the overall workflow of DEFRAG. Below, we detail its retrieval and generation optimizations.

\textbf{Edge-collaborative hybrid retrieval.} To broaden retrieval, DEFRAG extends the knowledge scope of edge-based RAG by sharing lightweight knowledge graphs with peer devices. To support breadth data coverage, these graphs are designed without dense vector representations. Following hybrid retrieval strategies~\cite{karpukhin2020dense}, the system adopts a two-stage process, further enhancing depth through collaboration among devices.

In the first stage, queries use coarse-grained sparse retrieval~\cite{lewis2020retrieval} over the shared knowledge graphs database (KGs DB) with term-based methods such as BM25 to identify the TOP-K relevant entities. In the second stage, the system locates the corresponding vector databases (Vector DBs) on peer edge devices based on the retrieved entities, fetching data chunks and their embeddings. If a target device is offline, retrieval proceeds to the next nearest device by proximity until the entity is found or retrieval is abandoned. Fine-grained dense retrieval~\cite{karpukhin2020dense} further reranks the TOP-K relevant entities, enhancing retrieval depth.

As shown in Table~\ref{tab:threshold_coverage}, relevance scores from dense retrieval can characterize query patterns. These scores are summarized into context features—including count, mean, standard deviation, maximum, and minimum for entities and data chunks—which are passed to the optimizer to support adaptive generation (see Section~\ref{subsection:modeling} for details).

\textbf{Optimizer-driven adaptive generation.} To support efficient and accurate generation, DEFRAG employs a multi-objective optimizer to select generation parameters for each query. These parameters include the choice of SLM size, quantization precision, and a combination of structured and unstructured RAG retrieval outputs. Structured information (entities and relationships) is extracted from shared knowledge graphs, while unstructured evidence (data chunks) is obtained from vector databases on local and peer devices.

The optimizer dynamically adjusts its strategy based on context features derived during retrieval, which reflect the structure and complexity of each query. For example, when a query involves multiple high-relevance entities, indicating a need for multi-entity reasoning, the optimizer selects higher-capacity models and richer structured input. In contrast, for a query focused on a single high-confidence entity, often corresponding to a single-hop question, it chooses smaller models and minimal input to improve efficiency.

Each candidate action receives a utility score balancing expected accuracy and efficiency, and the action with the highest score is selected. If the local device cannot execute it, the request is forwarded to the nearest capable peer device; otherwise, generation is performed locally. All interactions are logged to support online optimizer updates and ongoing refinement of the accuracy and efficiency trade-off.

By jointly optimizing retrieval and generation, DEFRAG leverages heterogeneous edge resources to democratize LLM services through decentralized collaboration among the crowd.

\subsection{Contextual Multi-Objective Multi-Armed Bandit Problem}
\label{subsection:modeling}

\begin{table}[t]
\centering
\footnotesize
\renewcommand{\arraystretch}{1.0}
\caption{Summary of frequently used notations.\label{table:notation_algorithms}}
\begin{tabularx}{\linewidth}{cX}
\toprule
\textbf{Notation} & \textbf{Explanation} \\
\midrule
$c_t$, $a_t$ & Context vector and action at query round $t$ \\
$\mathcal{C}$, $\mathcal{A}$ & Context space and combinatorial action space \\
$E^m(a)$ & Model efficiency of action $a$ (by model-cost) \\
$E^R(a)$ & Input efficiency of action $a$ (by input length) \\
$E(a)$ & Efficiency indicator of action $a$ \\
$ \alpha(c_t,a)$ & Accuracy indicator of action $a$ under context $c_t$ \\
$U(c_t,a)$ & Utility combining accuracy and efficiency \\
$\gamma$ & Weight of accuracy in the utility ($0\!\le\!\gamma\!\le\!1$) \\
$\bar{E}_{W}$ & Average efficiency over the latest $|W|$ queries \\
$\bar{\alpha}_{W}$ & Average accuracy over the latest $|W|$ queries \\
$\alpha_{\text{C}}$ & Accuracy of centralized SLMs/LLM baseline \\
\bottomrule
\end{tabularx}
\end{table}

To better understand and address the Multi-objective Optimizer, we formalize it as a Contextual Multi-Objective Multi-Armed Bandit (CMOMAB) problem tailored to our DEFRAG system architecture. Frequently used notations are summarized in Table~\ref{table:notation_algorithms} for reference.

\textbf{Context space.}
We formally define the context at each query round $t$ as:
$$
c_t := [n_t^e, \bar{c}_t^e, \tilde{c}_t^e, \hat{c}_t^e, \check{c}_t^e, n_t^d, \bar{c}_t^d, \tilde{c}_t^d, \hat{c}_t^d, \check{c}_t^d] \in \mathcal{C},
$$
where $n_t^e$ and $n_t^d$ are the numbers of entities and data chunks retrieved after initial sparse filtering. $\bar{c}_t^e$ and $\tilde{c}_t^e$ denote the mean and standard deviation of relevance scores for entities, while $\hat{c}_t^e$ and $\check{c}_t^e$ are their maximum and minimum scores. Analogously, $\bar{c}_t^d$, $\tilde{c}_t^d$, $\hat{c}_t^d$, and $\check{c}_t^d$ are the mean, standard deviation, maximum, and minimum relevance scores for selected data chunks. Here, $\mathcal{C}$ denotes the context space. The context design and its impact are evaluated through ablation studies in Section~\ref{subsection:ablation}.

\textbf{Action space.}
We define the action at each query round $t$ as a configuration consisting of four distinct parameters:
$$
a_t := [m_t, e_t, r_t, d_t] \in \mathcal{A} := \mathcal{M} \times \mathcal{E} \times \mathcal{R} \times \mathcal{D},
$$
where $\mathcal{M}$ denotes the set of selected SLMs with different sizes and quantization precision; $\mathcal{E}$, $\mathcal{R}$, and $\mathcal{D}$ are the TOP-K selected entities, relationships, and data chunks, respectively. The action space $\mathcal{A}$ is combinatorial, spanning all generation parameters (RAG and SLM configurations).

\textbf{Utility indicators.}
To select the most suitable action $a_t$ at each query round $t$, we use two utility indicators. The first is action accuracy. It is conditioned on the current query context $c_t$ and denoted as $\alpha(c_t, a_t)$. If the selected SLM’s response matches the reference answer, $\alpha(c_t, a_t) = 1$; otherwise, $\alpha(c_t, a_t) = 0$. However, it is computationally infeasible to exhaustively evaluate all actions by running each set of RAG generation parameters and comparing outputs. To address this, we use Bayesian online learning to estimate the probability of accuracy for each action given the current query context and update the model as outcomes are observed. This enables the system to adapt dynamically to changing query patterns (see Section~\ref{section:solution} for details). The model policy is therefore learned from observed context-action feedback rather than from model size, so a smaller or quantized model can be selected when its estimated accuracy is higher for the current query context.

The second indicator is the efficiency of each action $a$, denoted as $E(a)$, which reflects both the efficiency of the selected SLM and the chosen RAG parameters, and is defined below with higher value indicating greater efficiency:

\begin{equation}
E(a) = w_m E^{\text{m}}(a) + w_r E^{\text{R}}(a),
\label{eq:1}
\end{equation}
where $E^{\text{m}}(a) = 1/\kappa_m$, and $\kappa_m$ is an offline-profiled model-cost coefficient for the selected SLM configuration. The coefficient is computed from model memory footprint (size and quantization), FLOPs-level workload, and response latency.

The term $E^{R}(a) = (\mathcal{T}_{\max} - \mathcal{T}(a)) / \mathcal{T}_{\max}$ represents the input's efficiency coefficient, where $\mathcal{T}(a)$ is the input token length and $\mathcal{T}_{\max}$ is the maximum allowed. Larger RAG parameter values increase input length, which raises prefill and decode cost and latency~\cite{zhu2024accelerating}, thereby reducing $E^{R}(a)$.

The weights $w_m$ and $w_r$ (with $w_m + w_r = 1$) balance model and input efficiency. Higher $w_m$ discourages actions with high model-cost coefficients, while higher $w_r$ avoids long input lengths caused by large RAG parameters. These weights are dynamically adjusted online, allowing the system to balance model and input efficiency over time.

\textbf{Short-term objective.} At each query round $t$, the short-term objective is to select the action $a_t$ from the action space $\mathcal{A}$ that yields the highest utility for the current query context $c_t$. Utility combines accuracy $\alpha(c_t, a)$ and efficiency $E(a)$, aggregated by a weighted geometric mean~\cite{lan2010axiomatic} for a smooth and interpretable trade-off:

\begin{equation}
\max_{a \in \mathcal{A}} \quad \left[\alpha(c_t, a)\right]^\gamma \cdot \left[E(a)\right]^{1-\gamma}
\label{eq:2}
\end{equation}
Here, $\gamma \in [0, 1]$ is a weighting parameter that balances accuracy and efficiency, and is adjusted online to align with the long-term objective.

\textbf{Long-term objective.}
While the short-term objective balances accuracy and efficiency for each query, the long-term objective dynamically adjusts this trade-off to sustain performance. The aim is to keep average accuracy above a target threshold, measured over a sliding window, while maximizing efficiency. Formally:

\begin{equation}
\begin{aligned}
& \max \quad \lim_{|W|\to\infty} \bar{E}_{W} \\
& \text{s.t.} \quad \lim_{|W|\to\infty} \bar{\alpha}_{W} 
               \ge \alpha_{\text{C}} - \varepsilon
\end{aligned}
\end{equation}
where $\bar{\alpha}_{W} = \frac{1}{|W|} \sum_{t\in W} \alpha (c_t, a)$  is the average accuracy over the most recent $|W|$ queries, $\bar{E}_{W} = \frac{1}{|W|} \sum_{t\in W} E(a_t)$ is the average efficiency, $\alpha_{\text{C}}$ is the centralized SLMs or LLM baseline accuracy, and $\varepsilon$ is a tolerance constant.

The long-term objective guides the short-term trade-off by adaptively adjusting the weighting parameters $\gamma$, $w_m$, and $w_r$. When sliding window accuracy $\bar{\alpha}_{W}$ drops below the target $\alpha_{\text{C}}$, the system increases $\gamma$ to emphasize accuracy; if accuracy exceeds the target, $\gamma$ is decreased to prioritize efficiency. If the average input token count or model-cost indicator exceeds its predefined threshold, $w_m$ and $w_r$ are adjusted.

\section{Bayesian Online Learning Solution}
\label{section:solution}
In this section, we propose the \textbf{D}ecentralized \textbf{E}dge Collaboration \textbf{F}or \textbf{RAG} \textbf{B}ayesian \textbf{O}nline \textbf{L}earning (DEFRAG-BOL) algorithm to address the CMOMAB problem above. We also discuss deployment strategies to reduce computational overhead and analyze the regret bound of the algorithm.

\begin{algorithm}[t]
\caption{\textbf{D}ecentralized \textbf{E}dge Collaboration \textbf{F}or \textbf{RAG} \textbf{B}ayesian \textbf{O}nline \textbf{L}earning (DEFRAG-BOL)}
\label{alg:defrag-bol}
\SetAlgoLined
\DontPrintSemicolon

\textbf{Input:} Action space $\mathcal{A}$, total query rounds $T$. \\
\textbf{Hyperparams:} Kernel $k(\cdot, \cdot)$, exploration params $T_0, \beta_t$, window size $W_{\max}, K$, target accuracy $\alpha_C$, step sizes $\delta_\gamma, \delta_w$, efficiency limits $L_{\lim}, R_{\lim}$. \\
\textbf{Definitions:} Let $l_t$ denote the input token length and $r_t$ indicate model-cost usage at query round $t$. \\
\textbf{Initialize:} $\mathrm{GP} \leftarrow \emptyset$, history $W \leftarrow \emptyset$, weights $\gamma, w_m, w_r$. Compute initial $E(a)$ via Eq.~\ref{eq:1}.

\For{$t = 1,\dots,T$}{
    \textbf{Observe} query context $c_t$\;

    \If{$t \le T_0$}{
        \tcp{\textcolor{commentcolor}{Phase I: Random Exploration}}
        $a_t \sim \text{Unif}(\mathcal{A})$ \;
    }
    \Else{
        \tcp{\textcolor{commentcolor}{Phase II: UCB Exploitation}}
        \For{$a \in \mathcal{A}$}{
            $\mu_{t-1}, \sigma_{t-1} \gets \mathrm{GP}.\text{Predict}(c_t, a)$\;
            $\widehat{\alpha}(c_t,a) \gets \mu_{t-1} + \sqrt{\beta_t} \sigma_{t-1}$\;
            $U(c_t,a)\!\gets\![\widehat{\alpha}(c_t,a)]^{\gamma} \cdot E(a)^{1-\gamma}$\;
        }
        $a_t \gets \arg\max_{a\in\mathcal{A}} U(c_t,a)$ \;
    }

    \textbf{Execute} $a_t$, observe accuracy $\alpha_t$ and metrics $l_t$, $r_t$\;
    \textbf{Update} $W \leftarrow \{(c_\tau, a_\tau, \alpha_\tau, l_\tau, r_\tau)\}_{\tau = \max(1, t - W_{\max} + 1)}^{t}$\;
    \textbf{Update} $\mathrm{GP}$ posterior using data in $W$\;

    \If{$t > T_0$ \textbf{and} $t \pmod K \equiv 0$}{
        Compute means $\bar{\alpha}, \bar{L}, \bar{R}$ from window $W$\;

        \lIf{$\bar{\alpha} < \alpha_C$}{$\gamma \leftarrow \min(1, \gamma + \delta_{\gamma})$}
        \lElse{$\gamma \leftarrow \max(0, \gamma - \delta_{\gamma})$}

        \If{$\bar{L} > L_{\lim}$}{
            $w_r \leftarrow \min(1, w_r + \delta_w)$; $w_m \leftarrow 1 - w_r$
        }
        \ElseIf{$\bar{R} > R_{\lim}$}{
            $w_m \leftarrow \min(1, w_m + \delta_w)$; $w_r \leftarrow 1 - w_m$
        }
        \textbf{Update} $E(a)$ via Eq.~\ref{eq:1} with new $w_m, w_r$\;
    }
}
\end{algorithm}

\subsection{Algorithmic Framework}
\label{section:Algorithmic}
Initially, we explored linear-contextual bandit methods such as LinUCB, Hybrid UCB~\cite{li2010contextual}, and LinkUCBGlobal~\cite{cesa2013gang}, but these showed limited performance, converging to accuracies well below $\alpha_{\text{C}}$ due to non-linear relationships between contexts and actions. To address this, we adopt a Bayesian online learning approach, modeling the utility function with a Gaussian Process (GP)~\cite{williams2006gaussian} over the joint context-action space. As non-parametric estimators, GPs capture non-linear dynamics and provide uncertainty estimates, enabling a balance between exploration and exploitation.

\textbf{Gaussian process and kernel design.} We use Gaussian Processes as non-parametric approximators to estimate the accuracy component of the utility function, modeling each context-action pair $x \in \mathcal{X} = \mathcal{C} \times \mathcal{A}$ as a sample from $\mathcal{GP}(\mu(x), k(x, x'))$, where $\mu(x)$ is the mean and $k(x, x')$ is the kernel. The predictive uncertainty $\sigma(x)$, given by the kernel, reflects the variance of the GP prediction for each context-action pair. Based on observed non-linearities and varying smoothness in accuracy across different contexts and actions (Section~\ref{subsection:Data-Driven Analysis}), we adopt a stationary, anisotropic Matérn kernel (smoothness parameter 2.5), which fits the smooth yet non-linear nature of the utility function~\cite{srinivas2009gaussian}. To ensure stability, kernel hyperparameters are fixed during execution, selected through robustness testing on prior data~\cite{bull2011convergence}. Mathematical details are given in~\cite{williams2006gaussian}.

\textbf{Exploration and exploitation.} As delineated in Algorithm~\ref{alg:defrag-bol}, DEFRAG-BOL operates in two phases: exploration and exploitation.

In the \textit{exploration phase} ($t=1, \dots, T_0$), the system retrieves the context $c_t$ for each query and uniformly samples an action $a_t$ from the action space, executing it without utility prediction. The observed accuracy $\alpha(c_t, a_t)$ is recorded to update the GP posterior. Uniform sampling ensures unbiased exploration and avoids early exclusion of potentially optimal actions, which could hinder convergence~\cite{schaul2015prioritized}.

In the \textit{exploitation phase} ($t = T_0+1, \dots, T$), the system uses the GP posterior to predict the accuracy for each candidate action. Specifically, it computes the mean $\mu_{t-1}(c_t, a)$ and uncertainty $\sigma_{t-1}(c_t, a)$ to formulate an Upper Confidence Bound (UCB) estimate $\widehat{\alpha}(c_t,a)$. This accuracy estimate is integrated with the pre-computed efficiency $E(a)$ to derive a composite utility score $U(c_t, a) = [\widehat{\alpha}(c_t,a)]^{\gamma} \cdot E(a)^{1-\gamma}$. The optimizer selects the action maximizing this utility for execution. Upon observing the actual accuracy $\alpha_t$ and efficiency metrics ($l_t, r_t$), the history $W$ is updated as a sliding window retaining the most recent $W_{\max}$ samples. Its observed accuracy is then used to update the GP, allowing the model to adapt as new data become available.

To ensure long-term performance, the algorithm employs a \textit{periodic adaptation mechanism}. Every $K$ rounds, statistics $(\bar{\alpha}, \bar{L}, \bar{R})$ are computed over a window $W$ to recalibrate the weighting parameters. The accuracy weight $\gamma$ is increased when the moving average $\bar{\alpha}$ falls below the target $\alpha_C$, prioritizing accuracy, and decreased otherwise to improve efficiency. Meanwhile, the efficiency weights $w_r$ and $w_m$ are adjusted to improve resource utilization efficiency: $w_r$ is increased if the average input length $\bar{L}$ exceeds $L_{\lim}$, penalizing long contexts, while $w_m$ is raised when model-cost usage $\bar{R}$ exceeds $R_{\lim}$, discouraging high-cost model configurations. The efficiency table $E(a)$ is then recomputed to reflect these updated priorities in subsequent decisions.

\textbf{Practical issue and acceleration.}
A key limitation of Gaussian Processes is the $\mathcal{O}(N^3)$ time complexity for posterior computation, making real-time fitting impractical as $N$ grows~\cite{williams2006gaussian}, where $N$ is the number of historical data used for training. To improve efficiency, we adopt two strategies: (1) fit the GP model periodically instead of after every query, using fast $\mathcal{O}(N)$ prediction in between; and (2) restrict GP fitting to a fixed-size sliding window of recent queries, which also mitigates overfitting to outdated data. Together, these strategies reduce optimizer overhead while preserving the GP’s ability to model non-linear dependencies and maintain accuracy.

\subsection{Regret Bound Analysis}
\noindent Given Algorithm~\ref{alg:defrag-bol} and the short-term objective in Eq.~\ref{eq:2}, we derive the following theoretical results.

\begin{theorem}
Let $\delta\in(0,1)$, $\gamma\in(0,1)$, and define $\beta_t=2\log{(6|\mathcal{A}|/(\delta\pi^2t^2))}$. Assume the utility function over context-action pairs is sampled from a Gaussian Process with zero mean and Matérn kernel $k(\cdot,\cdot)$. Then, with probability at least $1-\delta$, the cumulative regret $R_T$ of Algorithm~\ref{alg:defrag-bol} after $T$ rounds satisfies
\[
R_T \leq \sqrt{ \frac{4E_{min}^{3-\gamma}\beta_T \lambda_I T }{ \log(1+\sigma^{-2}) } },
\]
where $E_{min} = \min_a E(a)$, $\lambda_I = \max_{c_t, a} I(\alpha(c_t,a);\hat{\alpha}(c_t,a))$ is the maximal mutual information gain, and $\sigma^2$ is the observation noise variance. In particular, for fixed hyperparameters, we have $R_T = O(\sqrt{T\lambda_I})$.
\end{theorem}

\begin{proof}
If a random variable $x\sim \mathcal{N}(0,1)$, then for $\theta>0$
\begin{align}
\Pr\{x>\theta\}&=1/\sqrt{2\pi}\int_\theta^{\infty}e^{-x^2/2}\\
&=1/\sqrt{2\pi}\int_0^{\infty}e^{-(x-\theta)^2/2-(x-\theta)\theta-\theta^2/2}\\
&\le e^{-\theta^2/2}\Pr\{x>0\}=e^{-\theta^2/2}/2.
\end{align}
Let $x=(\hat{\alpha}(c_t,a)-\mu_{t-1}(c_t,a))/\sigma_{t-1}(c_t,a)$ and $\theta=\sqrt{\beta_t}$, then it holds that
\begin{align}
\Pr\{|\hat{\alpha}(c_t,a)-\mu_{t-1}(c_t,a)|>\sqrt{\beta_t}\sigma_{t-1}(c_t,a)\}\le e^{-\beta_t/2}.\notag
\end{align} Apply union bound on both $a\in\mathcal{A}$ and $t\in T$ and set $\beta_t=2\log{(6|\mathcal{A}|/(\delta\pi^2t^2))}$. Then,
\begin{align}
|\hat{\alpha}(c_t,a)-\mu_{t-1}(c_t,a)|\le\sqrt{\beta_t}\sigma_{t-1}(c_t,a), \forall a, \forall t
\label{eq:9}
\end{align}holds with probability at least $1-\delta$.

According to line 17 in Algorithm~\ref{alg:defrag-bol}, it holds that
\begin{align}
&[\mu_{t-1}(c_t,a)+\sqrt{\beta_t}\sigma_{t-1}(c_t,a)]^\gamma[E(a)]^{1-\gamma}\\
\ge& [\mu_{t-1}(c_t,a^*)+\sqrt{\beta_t}\sigma_{t-1}(c_t,a^*)]^\gamma[E(a^*)]^{1-\gamma}\\
\ge& [\hat{\alpha}(c_t,a^*)]^\gamma[E(a^*)]^{1-\gamma}=U(c_t,a^*).
\end{align}where $E_{min}=\min_a\{E(a)\}$. The second inequality is by Eq.~\ref{eq:9}. Then, the regret $R_t$ is
\begin{align}
&R_t=U(c_t,a^*)-U(c_t,a)\\\notag
&=[\hat{\alpha}(c_t,a^*)]^\gamma[E(a^*)]^{1-\gamma}-[\hat{\alpha}(c_t,a)]^\gamma[E(a)]^{1-\gamma}\\ \notag
&\le ([\mu_{t-1}(c_t,a)+\sqrt{\beta_t}\sigma_{t-1}(c_t,a)]^\gamma-[\hat{\alpha}(c_t,a)]^\gamma)[E(a)]^{1-\gamma}\\ \notag
&\le (\mu_{t-1}(c_t,a)+\sqrt{\beta_t}\sigma_{t-1}(c_t,a)-\hat{\alpha}(c_t,a))[E(a)]^{1-\gamma}\\ \notag
&\le 2E_{min}^{1-\gamma}\sqrt{\beta_t}\sigma_{t-1}(c_t,a),
\end{align}
The second inequality is by the fact that $a^\gamma-b^\gamma<a-b$ when $\gamma\in(0,1)$ and $a>b>0$. Obviously, $\mu_{t-1}(c_t,a)+\beta_t\sigma_{t-1}(c_t,a)>\hat{\alpha}(c_t,a)$ holds by Eq.~\ref{eq:9}. The last inequality is also by Eq.~\ref{eq:9}. Since $\beta_t$ is a non-decreasing variable with $t$, it holds that
\begin{align}
&R_t^2=4E_{min}^{3-\gamma}{\beta_t}\sigma_{t-1}^2(c_t,a)\le4E_{min}^{3-\gamma}{\beta_T}\sigma^2\sigma^{-2}\sigma_{t-1}^2(c_t,a)\notag\\
&\le \frac{4E_{min}^{3-\gamma}{\beta_T}}{\log{(1+\sigma^{-2})}}\log(1+\sigma^{-2}\sigma_{t-1}^2(c_t,a)).
\end{align}
The last inequality holds by $\sigma^{-2}\sigma_{t-1}^2(c_t,a)\log(1+\sigma^{-2})\le \sigma^{-2}\log(1+\sigma^{-2}\sigma_{t-1}^2(c_t,a))$. Due to the information theory~\cite{cover1999elements}, the information gain about $\hat{\alpha}(c_t,a)$ for the decision $a\in\mathcal{A}$ can be measured by the mutual information $I({\alpha}(c_t,a);\hat{\alpha}(c_t,a))$ between $\hat{\alpha}(c_t,a)$ and the actual observation ${\alpha}(c_t,a)$. Assuming ${\alpha}(c_t,a)-\hat{\alpha}(c_t,a)\sim\mathcal{N}(0,\sigma^2)$, it holds by induction that
\begin{align}
\hspace{-2pt}I({\alpha}(c_t,a);\hat{\alpha}(c_t,a))=\frac{1}{2}\sum_{t=1}^T\log(1+\sigma^{-2}\sigma_{t-1}^2(c_t,a)).
\end{align}
Further, by the Cauchy-Schwarz inequality, we have
\begin{align}
\sum\nolimits_{t=1}^TR_t\le\sqrt{T\sum\nolimits_{t=1}^TR_t^2}\le\sqrt{\frac{4E_{min}^{3-\gamma}{\beta_T}\lambda_IT}{\log{(1+\sigma^{-2})}}},
\end{align}
where $\lambda_I=\max_{c_t,a}{I({\alpha}(c_t,a);\hat{\alpha}(c_t,a))}$ and is upper bounded by $O(T^{d(d+1)/(2v+d(d+1))}\log{T})$ with Matérn kernals according to~\cite{srinivas2009gaussian}. The cumulative regret of Algorithm~\ref{alg:defrag-bol} is upper bounded by $O(\sqrt{T\lambda_I})$, which completes the proof.
\end{proof}

\section{Experimental Evaluation}
\label{section:evaluation}
\subsection{Experimental Setup}

\begin{figure}[t]
\centering
\includegraphics[width=0.9\columnwidth]{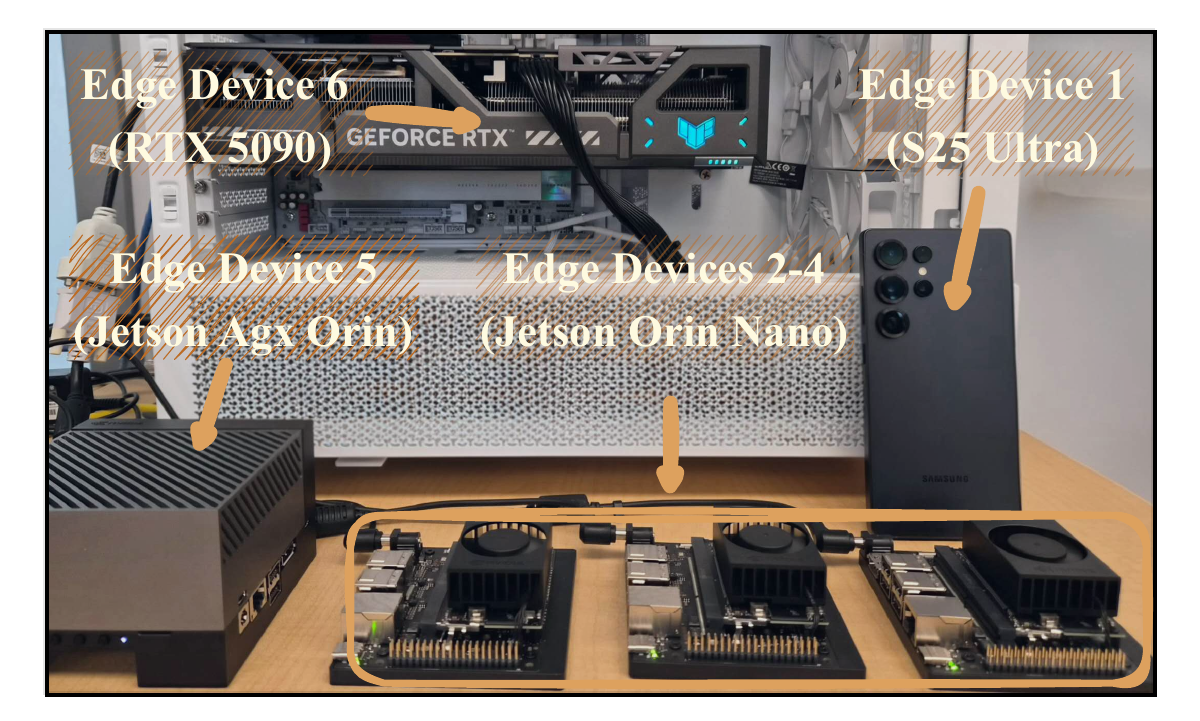}
\caption{DEFRAG prototype testbed.}
\label{fig:prototype_design}
\end{figure}

Our prototype testbed, shown in Figure~\ref{fig:prototype_design}, comprises six heterogeneous edge devices to emulate realistic edge deployment scenarios. Edge Device 1 is a Samsung S25 Ultra smartphone, representing mobile clients with limited resources. Edge Devices 2 to 4 are NVIDIA Jetson Orin Nano boards, emulating lightweight inference nodes such as laptops or AR glasses. Edge Device 5 is an NVIDIA Jetson AGX Orin, representing higher-capacity platforms typical of autonomous robotics. Edge Device 6 is a workstation with an NVIDIA RTX 5090 GPU, emulating a cellular base station or edge server.

All devices are connected via Wi-Fi, allowing evaluation of network bottlenecks and peer-to-peer communication. In DEFRAG, each device acts both as server and client: as a server, it hosts a local vector database and provides SLM inference services via Flask-based RESTful APIs~\cite{flask-docs} and Ollama APIs~\cite{ollama-api-docs}; as a client, it accesses data and computational resources from peer devices. To simulate user-specific knowledge preferences and non-uniform data distribution, we partition Wikipedia pages from each of the NQ and HotpotQA datasets into ten balanced clusters using TF-IDF vectors and a balanced K-means algorithm. Each of the five peer devices stores two clusters, while Device 6 stores all clusters and serves as a fallback for retrieval.

The action space is defined by $|\mathcal{M}|=12$ and $|\mathcal{E}|=|\mathcal{R}|=|\mathcal{D}|=11$, yielding $|\mathcal{X}|=12 \times 11^3 \approx 1.6 \times 10^4$ possible actions. Given the additional variability from different query contexts, such complexity underscores the necessity of a data-efficient learning mechanism. Initial efficiency weights $w_m$ and $w_r$ are set to $0.5$, as there is no strong prior evidence favoring either model size or input length in terms of cost efficiency. Given that accuracy is the primary long-term objective, we initialize the accuracy weight $\gamma$ to $0.8$. These weights are updated every 20 queries based on recent performance. The Gaussian Process model is also updated every 20 rounds, using up to the latest 1500 samples to balance computational overhead and prediction accuracy. Unless specified, lines show cumulative averages up to each round, and shaded areas indicate the difference between cumulative and moving averages over the most recent 500 rounds.

\subsection{Experimental Result}

We evaluate DEFRAG and all baseline methods on the heterogeneous prototype testbed. For the single-hop NQ dataset, baseline SLMs and LLMs are evaluated with standard RAG; for the multi-hop HotpotQA dataset, GraphRAG is used to enable relational reasoning, as explained in Section~\ref{subsection:Data-Driven Analysis}. Each method is assessed on 500 QA pairs using a well-trained optimizer (with analysis provided in the following section). Evaluation metrics include answer accuracy, total cost (covering inference and network overhead) to reflect cost efficiency, and system scalability.

\begin{figure}[t]
\centering
\includegraphics[width=\linewidth]
{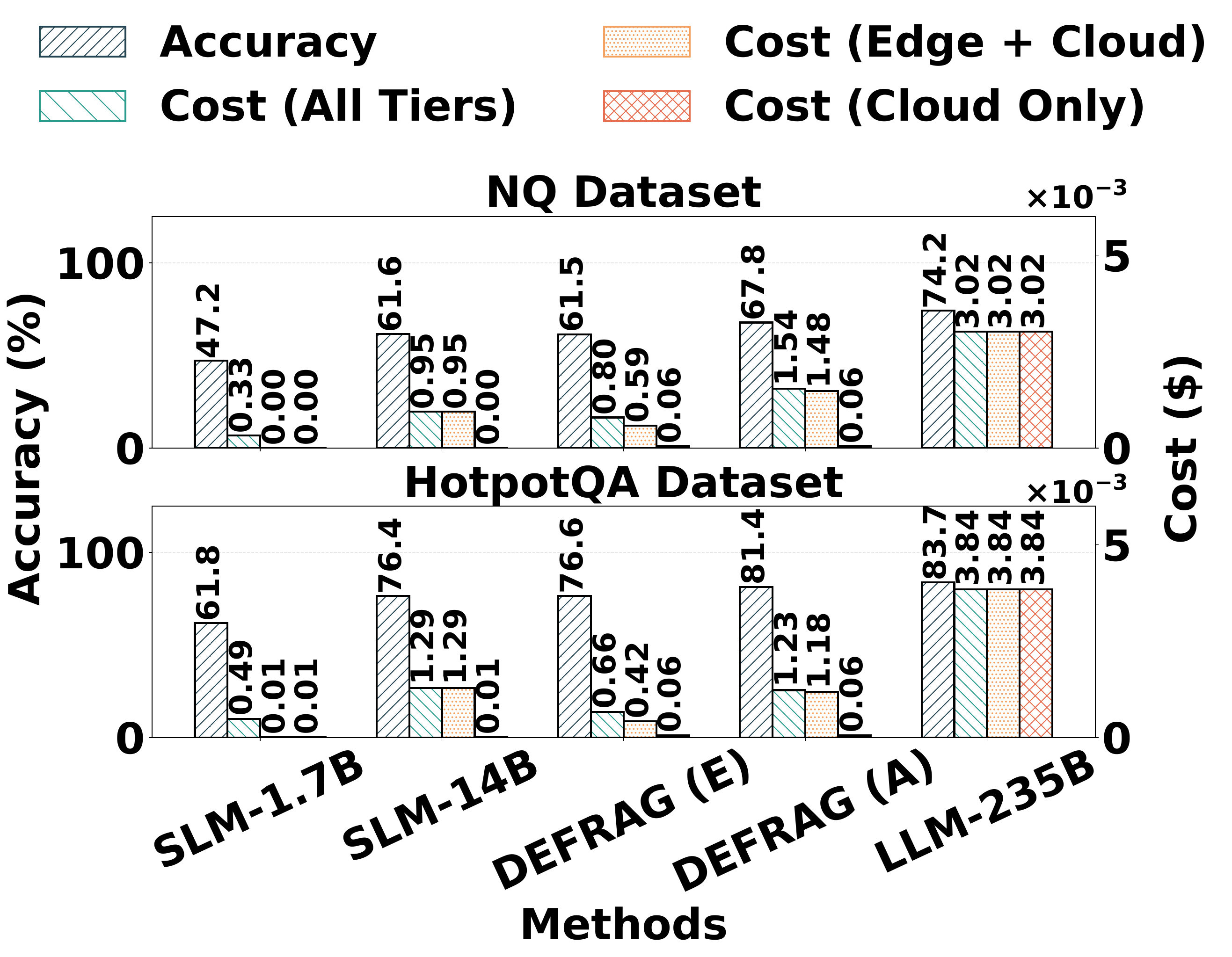}
\caption{Comparison of accuracy and cost between DEFRAG and centralized SLM/LLM baselines augmented with RAG. The top panel presents results on the NQ dataset using standard RAG, while the bottom panel shows results on the HotpotQA dataset using GraphRAG.}
\label{fig:performance_comparison}
\end{figure}

\begin{figure}[t]
\centering
\includegraphics[width=\linewidth]{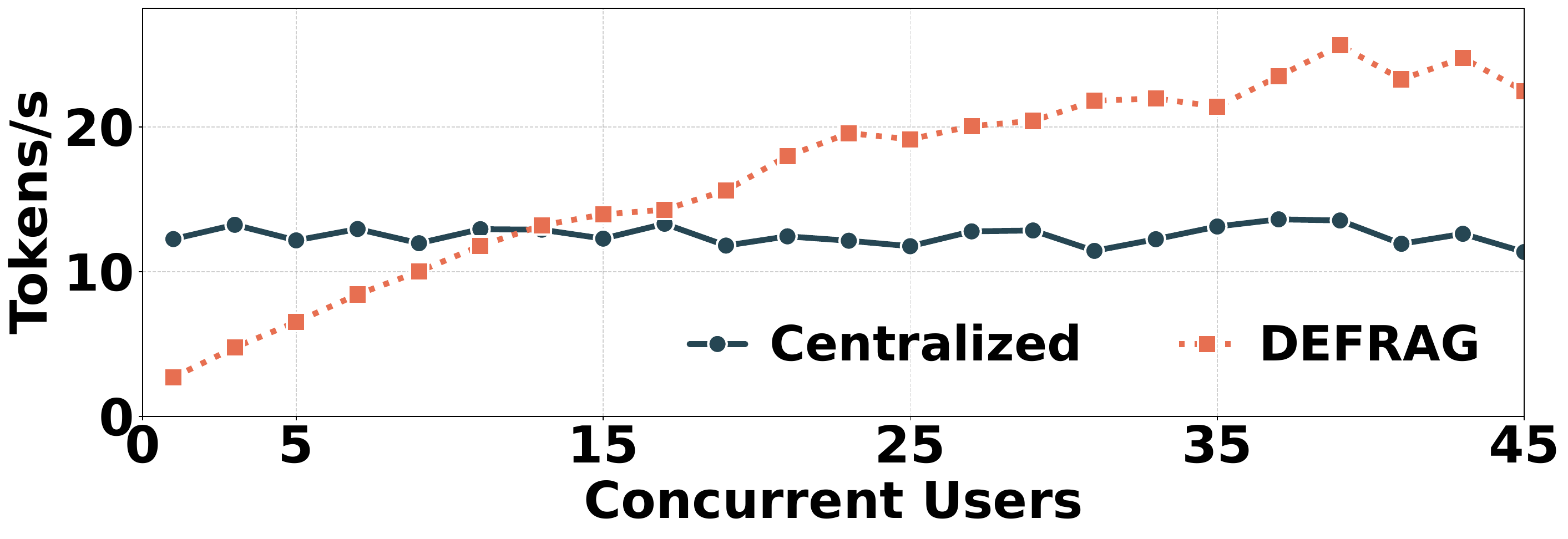}
\caption{Comparison of average throughput (tokens/s) under varying numbers of concurrent users for centralized and our DEFRAG (decentralized) system architecture.}
\label{fig:scalability_comparison}
\end{figure}

Figure~\ref{fig:performance_comparison} compares DEFRAG with static and centralized SLM and LLM baselines on both datasets, evaluating both accuracy and total cost. Baselines use SLMs and LLM with corresponding RAG parameters, and accuracy is measured accordingly. Total cost includes model inference (calculated using official Qwen API pricing\footnote{\url{https://www.alibabacloud.com/help/en/model-studio/models}}) and network overhead (estimated from AWS public data transfer rates\footnote{\url{https://aws.amazon.com/ec2/pricing/on-demand/}}). For SLM-1.7B and SLM-14B baselines, inference is performed locally on the edge server, with costs including computation and user-to-device data transfer. The LLM-235B baseline is cloud-hosted, with similar cost components. For DEFRAG, costs include computation, user-to-device data transfer, and network transfer for coordination and peer data exchange.

\begin{figure*}[t]
\centering
\includegraphics[width=1\linewidth]{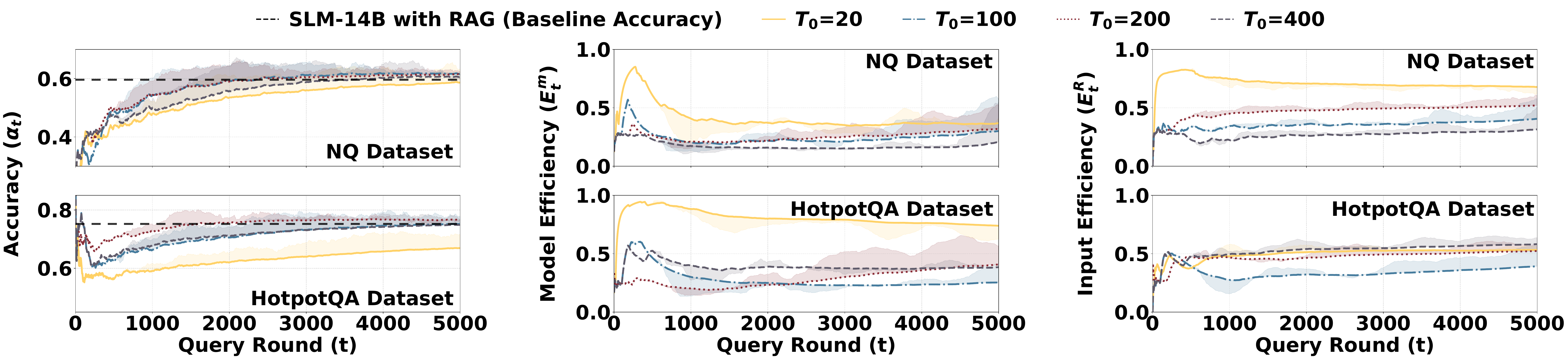}
\caption{Comparison of DEFRAG accuracy and efficiency across initial exploration steps ($T_0$) on NQ and HotpotQA datasets.}
\label{fig:T_0_comparision}
\end{figure*}

Experimental results show that under the accuracy-oriented configuration (DEFRAG (A), $\alpha_{\text{C}}=0.99$), the accuracy gap between edge-based SLMs and cloud-based LLMs is significantly reduced. On two benchmark QA datasets, edge accuracy improves from 61.5\% and 76.6\% (SLM-14B) to 67.8\% and 81.4\%, narrowing the gap with LLM-235B from 12.7\% and 7.1\% to 6.4\% and 2.3\%, respectively. For cost, DEFRAG achieves reductions of up to 98.4\% compared to LLM-235B when mutual resource sharing is enabled among peer devices. In cost-oriented mode, DEFRAG (E) maintains accuracy comparable to SLM-14B baseline, while reducing total cost by up to 48.8\% when all device models (1.7B to 235B) are included, and by 67.4\% when considering only the edge and cloud server models (14B and 235B).

As shown in Figure~\ref{fig:scalability_comparison}, DEFRAG demonstrates strong scalability as concurrency increases. While the centralized baseline on a single edge server quickly saturates and incurs queuing delays with more users, DEFRAG distributes retrieval and inference across heterogeneous devices, supporting near-linear throughput growth. In our experiments, DEFRAG achieves up to 97.8\% higher peak throughput than the centralized approach, with further gains as more edge devices join.

These results confirm that while SLMs still lag behind LLMs in accuracy with RAG, DEFRAG substantially narrows this gap, reduces costs, and improves scalability by leveraging heterogeneous edge resources to provide decentralized and democratized LLM services from cloud to crowd.

\subsection{Ablation Studies}
\label{subsection:ablation}

We first evaluate the effect of different initial exploration steps ($T_0$) on the convergence of DEFRAG in terms of accuracy and efficiency. Experiments are conducted over 5,000 query rounds (10 iterations per dataset). Due to dynamic adjustment of efficiency weights, we report model efficiency ($E_t^m$) and input efficiency ($E_t^R$), as shown in Figure~\ref{fig:T_0_comparision}.

Among the tested settings, model efficiency peaked at $T_0 = 20$ for both datasets, as the system favored smaller models, resulting in lower accuracy and limited convergence. This illustrates the risk of insufficient exploration, which may introduce bias and hinder convergence. At $T_0 = 100$, accuracy approached that of the cloud-scale SLM baseline (SLM-14B), though efficiency could still improve. For $T_0 \geq 200$, both accuracy and efficiency converged. After about 2,000 iterations, accuracy stabilized, while model selection gradually shifted toward smaller models to further improve efficiency. Ultimately, model efficiency converged with continued iterations, balancing accuracy and resource use. Based on these results, we select $T_0 = 200$ as the final setting.

Notably, NQ accuracy increased steadily, while HotpotQA showed greater fluctuations with varying exploration steps. This is because HotpotQA relies more on common-sense knowledge, allowing even partial relevant information to support correct answers~\cite{groeneveld2020simple}. In contrast, NQ queries often involve obscure facts outside pre-training data, making them more dependent on retrieval~\cite{jin2024long}, and resulting in a lower baseline accuracy than HotpotQA.

\begin{table}[!t]
    \centering
    \footnotesize
    \renewcommand{\arraystretch}{1.0}
    \caption{Comparison of average accuracy and efficiency across context features ($c$) on HotpotQA and NQ datasets.}
    \begin{tabularx}{\linewidth}{lXXX}
        \toprule
        \textbf{Feature set $c$} & $\overline{\alpha}$ (\%) & $\overline{E}^{\!m}$ & $\overline{E}^{\!R}$ \\
        \midrule
        \multicolumn{4}{c}{NQ Dataset} \\
        \midrule
        Mean+Std                  & 60.3 & 0.24 & 0.60 \\
        Mean+Std+Num            & 60.8 & 0.28 & 0.55 \\
        \textbf{Mean}+\textbf{Std}+\textbf{Num}+\textbf{Min}+\textbf{Max} &\textbf{ 61.5} & \textbf{0.32} & \textbf{0.52} \\
        \midrule
        \multicolumn{4}{c}{Hotpot Dataset} \\
        \midrule
        Mean+Std                  & 72.9 & 0.28  & 0.55  \\
        Mean+Std+Num            & 75.4 & 0.19  & 0.42  \\
        \textbf{Mean}+\textbf{Std}+\textbf{Num}+\textbf{Min}+\textbf{Max} & \textbf{76.6} & \textbf{0.41}  & \textbf{0.52}  \\
        \bottomrule
    \end{tabularx}
    \label{tab:context_comparison}
\end{table}

\begin{figure}[!t]
\centering
\includegraphics[width=\linewidth]{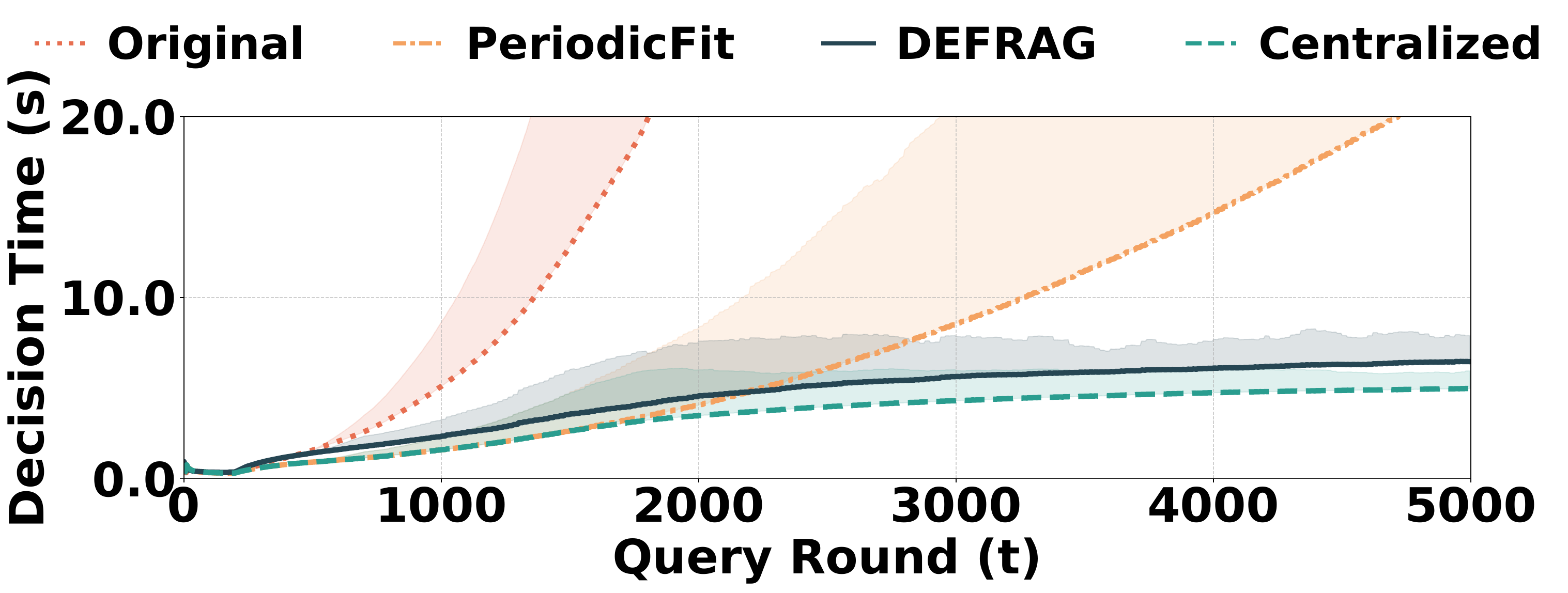}
\caption{Comparison of decision delay (retrieval plus optimizer prediction) across different optimization strategies.}
\label{fig:context_comparision}
\end{figure}

Table~\ref{tab:context_comparison} shows how different context feature sets affect average accuracy and efficiency on the NQ and HotpotQA datasets. By incorporating features that capture query structure and complexity, the optimizer can dynamically adjust its actions. Richer representations of relevance score distributions help the optimizer better interpret queries and select actions. While mean and standard deviation reflect central tendency and dispersion, they may not distinguish between a few strong matches and many moderate ones. Including count (num) indicates the density of relevant information, while adding min and max reveals the score range and the presence of outliers.

\begin{table*}[!t]
    \centering
    \footnotesize
    \renewcommand{\arraystretch}{1.0}
    \caption{HotpotQA 500-query system comparison.}
    \begin{tabular*}{\textwidth}{@{\extracolsep{\fill}}lcccccl@{}}
        \toprule
        \textbf{Setting} & \textbf{RAG parameters} & \textbf{Model policy} & \textbf{Accuracy} & \textbf{Latency} & \textbf{Norm. cost} & \textbf{Description} \\
        \midrule
        DEFRAG(A) & Adaptive & Edge adaptive & 81.4\% & 6.233s & 0.292 & Accuracy-focused setting \\
        DEFRAG(E) & Adaptive & Edge adaptive & 76.6\% & 7.490s & 0.158 & Efficiency-focused setting \\
        Fixed-CollabRAG & Fixed & Local 14B & 77.2\% & 3.322s & 0.312 & Fixed retrieval, local generation \\
        Adaptive Retrieval & Adaptive & Local 14B & 79.6\% & 3.057s & 0.305 & Adaptive RAG parameters only \\
        Adaptive Model & Fixed & Edge adaptive & 78.0\% & 6.102s & 0.295 & Adaptive model route only \\
        Heuristic-Fastest & Fixed & Rule-based & 60.8\% & 5.577s & 0.078 & Fastest feasible route \\
        Heuristic-Strongest & Fixed & Rule-based & 84.2\% & 6.355s & 1.000 & Strongest feasible route \\
        \bottomrule
    \end{tabular*}
    \label{tab:revision_task2_system_rows}
\end{table*}

Figure~\ref{fig:context_comparision} compares decision delay, including retrieval and optimizer prediction time, across different acceleration strategies. Without optimization (Original), decision delay is significantly higher, exhibiting $\mathcal{O}(N^3)$ complexity. Periodic GP fitting (PeriodicFit) mitigates delay to some extent but does not eliminate the exponential growth. Combining periodic GP fitting with window restriction (DEFRAG) substantially reduces decision delay, resulting in near-linear time complexity.

The centralized baseline (Centralized) measures delay on an edge server without network overhead. DEFRAG achieves a similar delay, despite additional network latency from distributed peer database access, indicating that network communication is not a bottleneck under stable testbed conditions.

Table~\ref{tab:revision_task2_system_rows} compares seven HotpotQA settings under the same prototype deployment. Each setting uses 500 queries and reports retrieval-plus-generation latency. Fixed-CollabRAG uses fixed RAG parameters and a fixed 14B model. Adaptive Retrieval changes only the RAG parameters. Adaptive Model changes only model selection. Both improve over Fixed-CollabRAG, but both remain below DEFRAG(A). This shows that RAG parameter selection and model selection are both needed. Compared with Fixed-CollabRAG, DEFRAG(A) raises accuracy from 77.2\% to 81.4\% and lowers normalized cost from 0.312 to 0.292. DEFRAG(E) keeps similar accuracy and lowers normalized cost to 0.158. Heuristic-Fastest and Heuristic-Strongest show the two single-objective choices: lower cost with lower accuracy, or higher accuracy with higher cost. Table~\ref{tab:revision_task2_system_rows} therefore shows how DEFRAG balances accuracy, latency, and cost.

\subsection{Robustness and Generalization Analysis}

We also test DEFRAG under four controlled settings that affect edge-collaborative RAG: model loading, mobile route instability, skewed data placement, and another QA domain. These tests examine whether DEFRAG remains effective when model execution, routes, data placement, or the task domain changes. The prototype setting is the reference, and each experiment changes one factor at a time.

\begin{figure}[t]
\centering
\includegraphics[width=\linewidth]{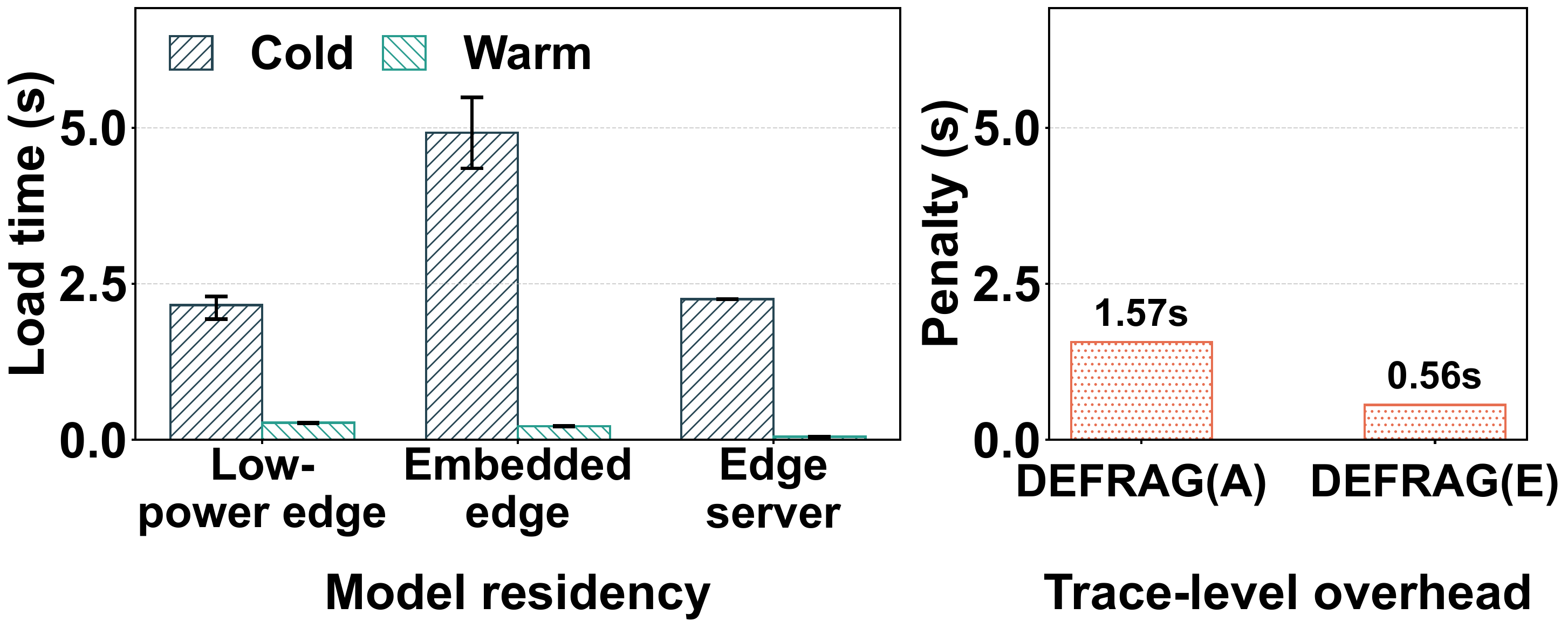}
\caption{Model residency and switching overhead across serving-peer routes.}
\label{fig:revision_model_overhead}
\end{figure}

\begin{figure*}[!t]
\centering
\begin{subfigure}[t]{0.33\textwidth}
    \centering
    \includegraphics[width=\linewidth]{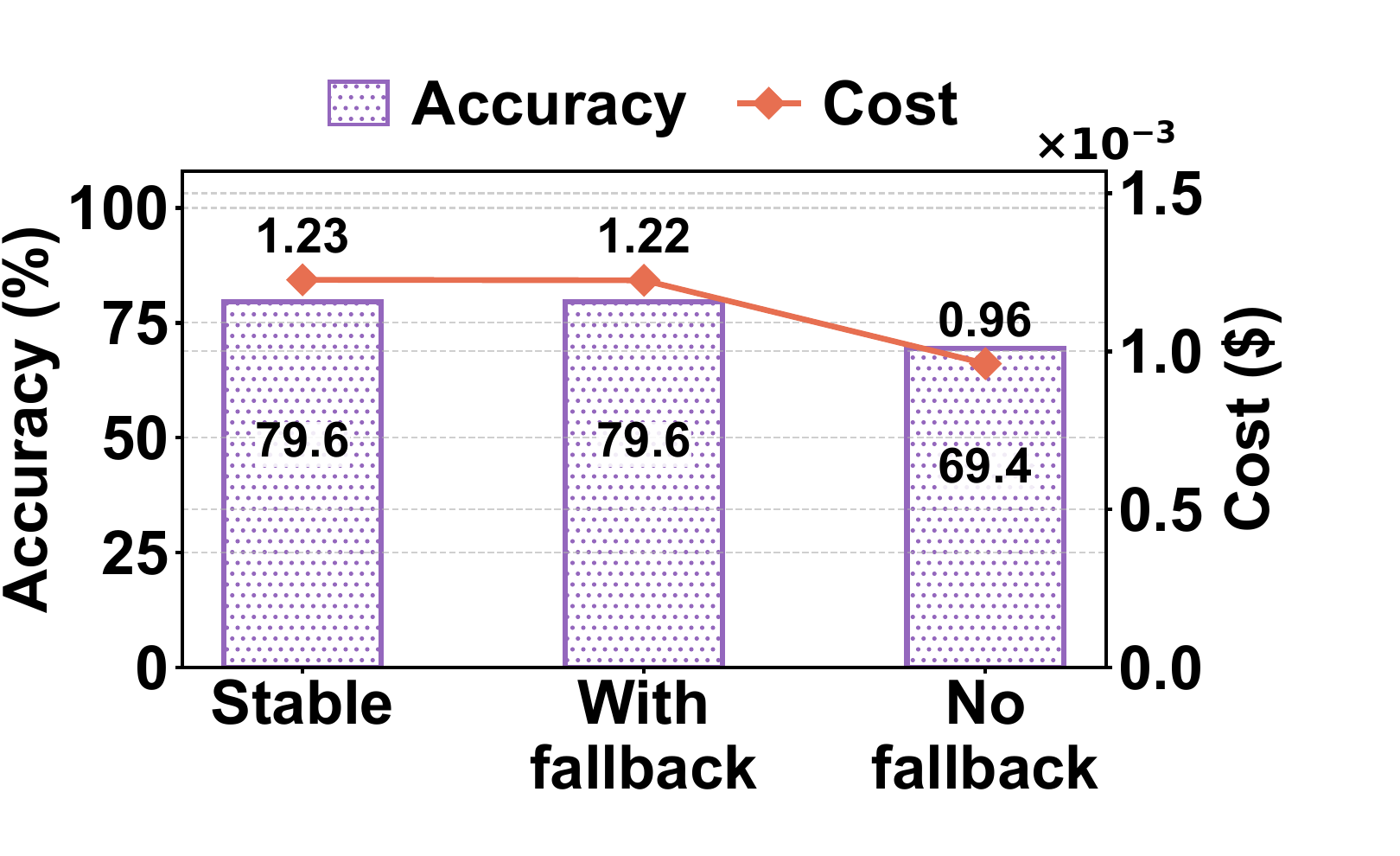}
\end{subfigure}\hfill
\begin{subfigure}[t]{0.33\textwidth}
    \centering
    \includegraphics[width=\linewidth]{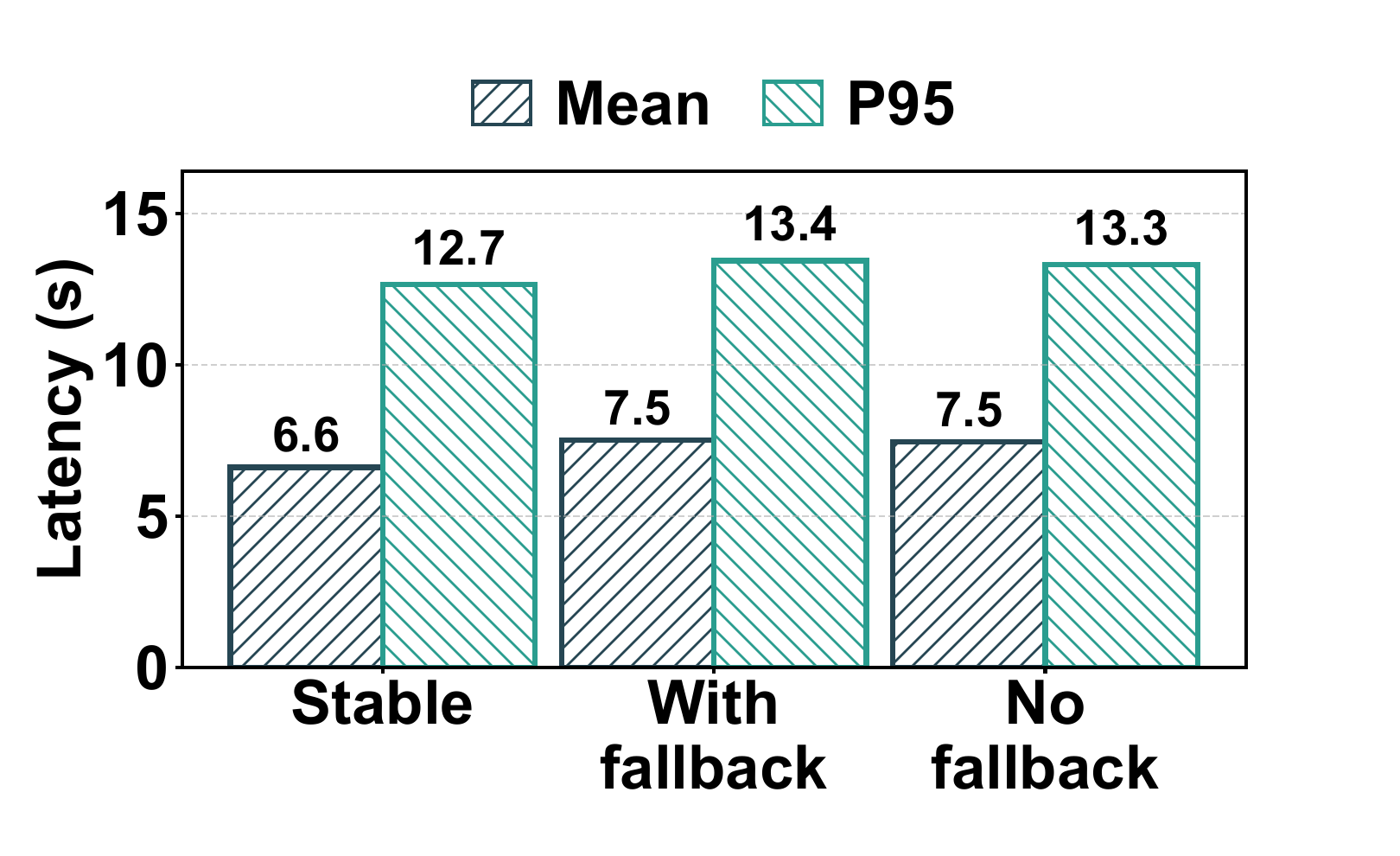}
\end{subfigure}\hfill
\begin{subfigure}[t]{0.33\textwidth}
    \centering
    \includegraphics[width=\linewidth]{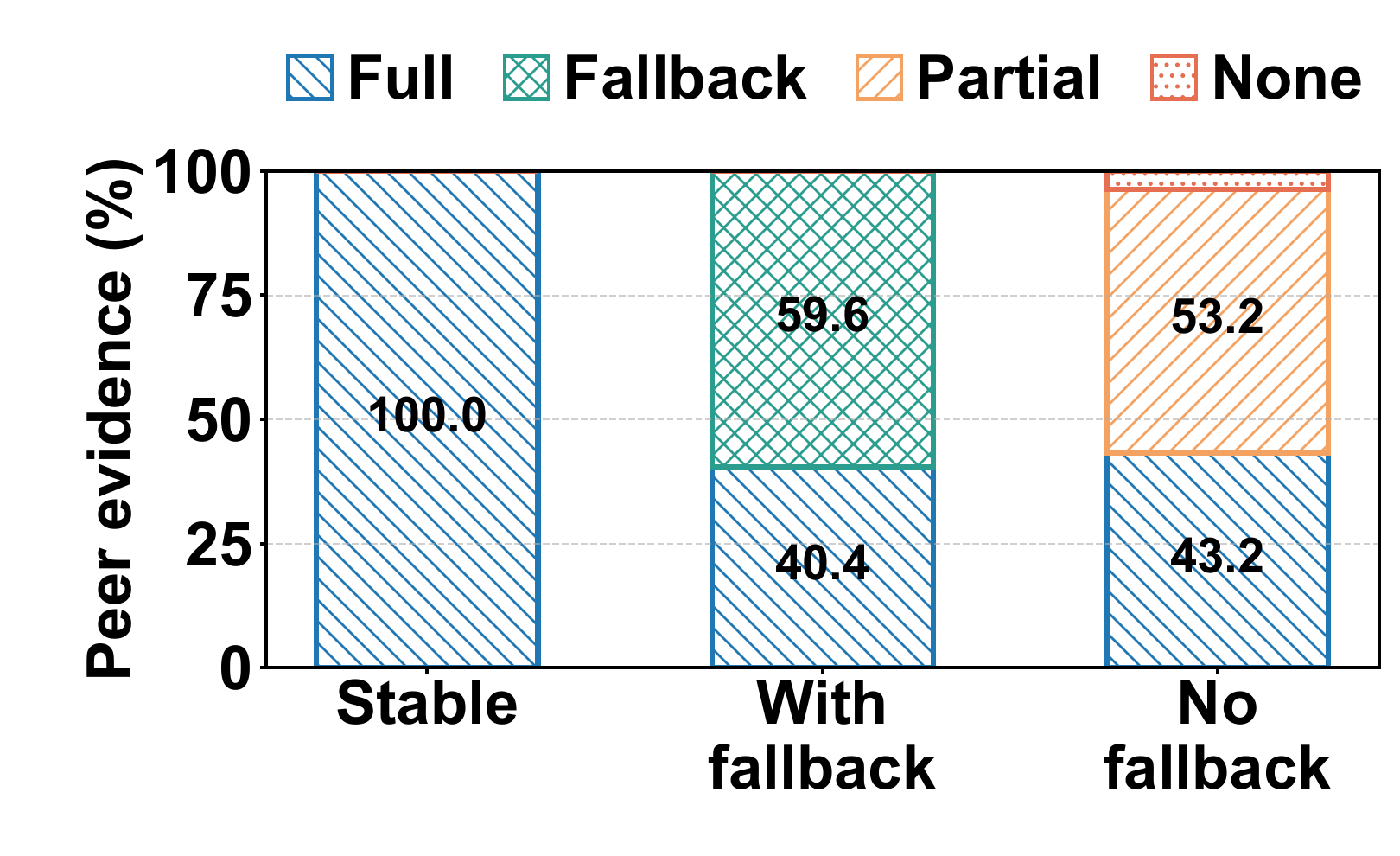}
\end{subfigure}
\caption{Paired mobile route stress on HotpotQA: accuracy/cost, latency, and evidence availability.}
\label{fig:revision_comm_churn}
\end{figure*}

\begin{figure}[t]
\centering
\includegraphics[width=0.9\columnwidth]{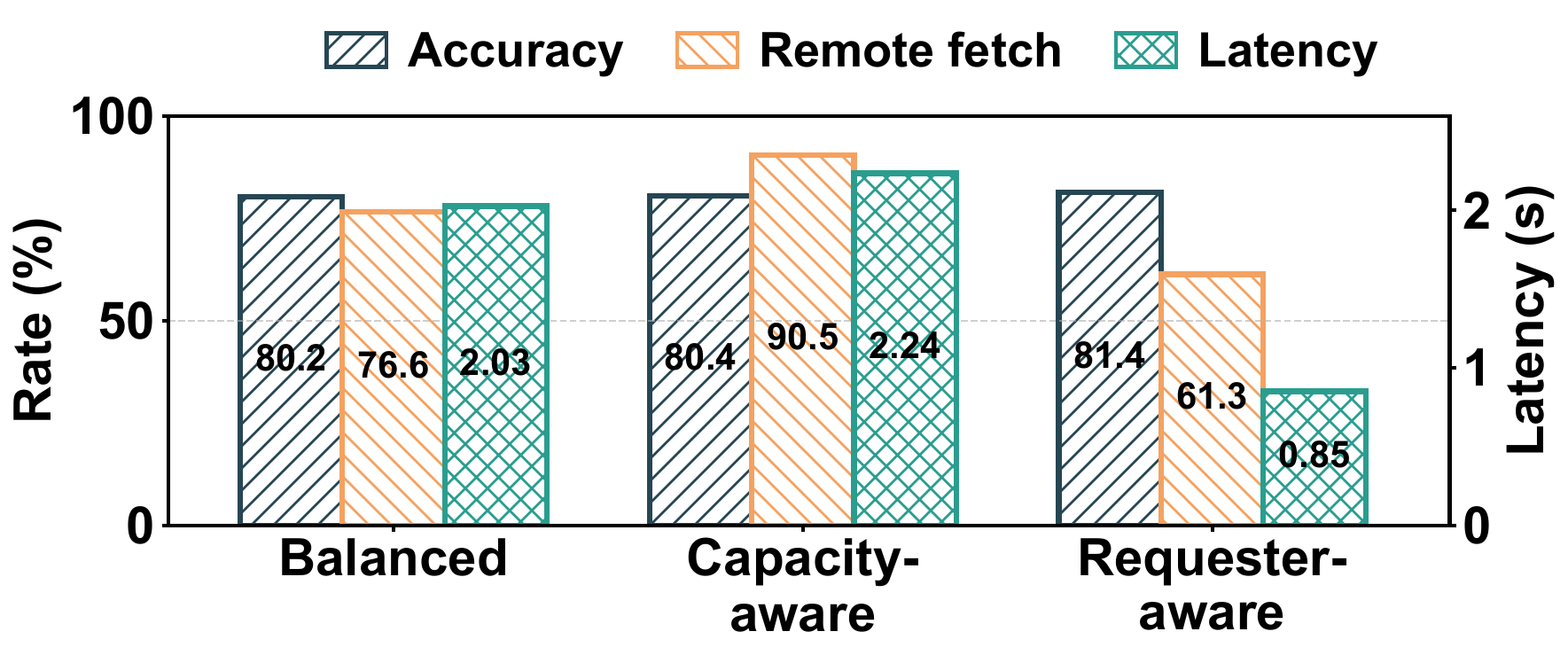}
\caption{Impact of skewed data placement on accuracy, remote fetch, and retrieval~latency.}
\label{fig:revision_skew}
\end{figure}

Figure~\ref{fig:revision_model_overhead} quantifies model loading and switching overhead in dynamic model selection. Figure~\ref{fig:revision_model_overhead}(a) profiles cold and warm execution with HotpotQA retrieval prompts on three route classes. Low-power edge routes represent lightweight peers such as Jetson Orin Nano boards and mobile terminals. Embedded edge routes represent higher-capacity embedded AI platforms, such as Jetson AGX Orin-class devices used in robotics. Edge-server routes represent workstation-class edge servers such as the RTX 5090 node. Cold execution loads a non-resident model, while warm execution reuses a resident or cached model. Cold loading takes 2.16\,s on low-power edge routes, 4.92\,s on embedded edge routes, and 2.25\,s on edge-server routes. Warm reuse stays below 0.28\,s in all three classes. This shows that model residency has a direct effect on action cost.

Figure~\ref{fig:revision_model_overhead}(b) replays the HotpotQA action traces of DEFRAG(A/E) with the measured loading and switching penalties. Loading and switching add 1.566\,s/query to DEFRAG(A), or 23.09\% of retrieval-plus-generation latency. They add 0.560\,s/query to DEFRAG(E), or 6.39\%. The lower penalty of DEFRAG(E) comes from more reuse of resident models, with a 91.4\% warm-resident ratio compared with 69.3\% for DEFRAG(A). This result shows that dynamic model selection has measurable I/O cost, and that efficiency-oriented selection can reduce this cost by reusing resident or cached models more often.

Figure~\ref{fig:revision_comm_churn} evaluates DEFRAG under mobile route instability. We use paired HotpotQA runs with the same queries, balanced placement, and DEFRAG(A) action sequence. The Stable run has no injected route loss. In the stress runs, the requester is a mobile device, and its local shard remains available. We inject a 30\% timeout/drop rate only on non-local peer data and model routes. In the fallback setting, the edge server serves as a redundant data and model route. In the no-fallback setting, the system skips missing peer evidence, and an unavailable model route is replaced by local 1.7B generation.

Figure~\ref{fig:revision_comm_churn}(a) reports accuracy and per-query cost. The fallback setting keeps the same 79.6\% accuracy as Stable, with similar cost (\$0.001225 vs. \$0.001226 per query). The no-fallback setting has lower cost, but the reduction comes from degraded service: missing peer evidence is skipped, local 1.7B generation is used more often, and accuracy drops to 69.4\%. Figure~\ref{fig:revision_comm_churn}(b) shows that fallback increases mean retrieval-plus-generation latency from 6.618\,s to 7.513\,s. This is the measured latency cost of route recovery.

Figure~\ref{fig:revision_comm_churn}(c) explains why accuracy is preserved. With fallback, all queries keep full peer evidence: 40.4\% through the original peer routes and 59.6\% through edge-server fallback. Without fallback, only 43.2\% of queries keep full peer evidence. Another 53.2\% keep partial peer evidence, and 3.6\% receive no peer evidence. In the same no-fallback run, 30.4\% of queries also use local 1.7B generation after a model-route failure. These results show that edge-server redundancy preserves evidence completeness and model-route continuity under mobile-like route stress.

Figure~\ref{fig:revision_skew} evaluates non-uniform data placement. We compare three placements. Balanced uses the original group assignment. Capacity-aware places frequently requested evidence groups on higher-capacity peers. Requester-aware places evidence groups closer to the requester origins that most often need them. All runs use the same HotpotQA queries and the same DEFRAG(A) selected actions. Thus, the comparison isolates the effect of data placement.

\begin{table}[!t]
    \centering
    \footnotesize
    \renewcommand{\arraystretch}{1.0}
    \setlength{\tabcolsep}{5pt}
\caption{Comparison of accuracy and normalized cost on the HarryPotter QA task.}
\begin{tabular*}{\columnwidth}{@{\extracolsep{\fill}}lcccc@{}}
    \toprule
        \textbf{Setting} & \textbf{RAG} & \textbf{Model} & \textbf{Accuracy} & \textbf{Cost} \\
        \midrule
        No retrieval & None & Local 14B & 40.6\% & 0.0680 \\
        RAG & Fixed & Local 14B & 47.2\% & 0.2290 \\
        GraphRAG & Fixed & Local 14B & 56.0\% & 0.4062 \\
        DEFRAG(E) & Adaptive & Edge adaptive & 47.8\% & 0.0341 \\
        DEFRAG(A) & Adaptive & Edge adaptive & 53.4\% & 0.0802 \\
        Cloud GraphRAG & Fixed & Cloud 235B & 62.2\% & 1.0000 \\
        \bottomrule
    \end{tabular*}
    \label{tab:revision_hp_comparison}
\end{table}

\begin{figure}[t]
\centering
\begin{subfigure}[t]{0.608\linewidth}
    \centering
    \includegraphics[width=\linewidth]{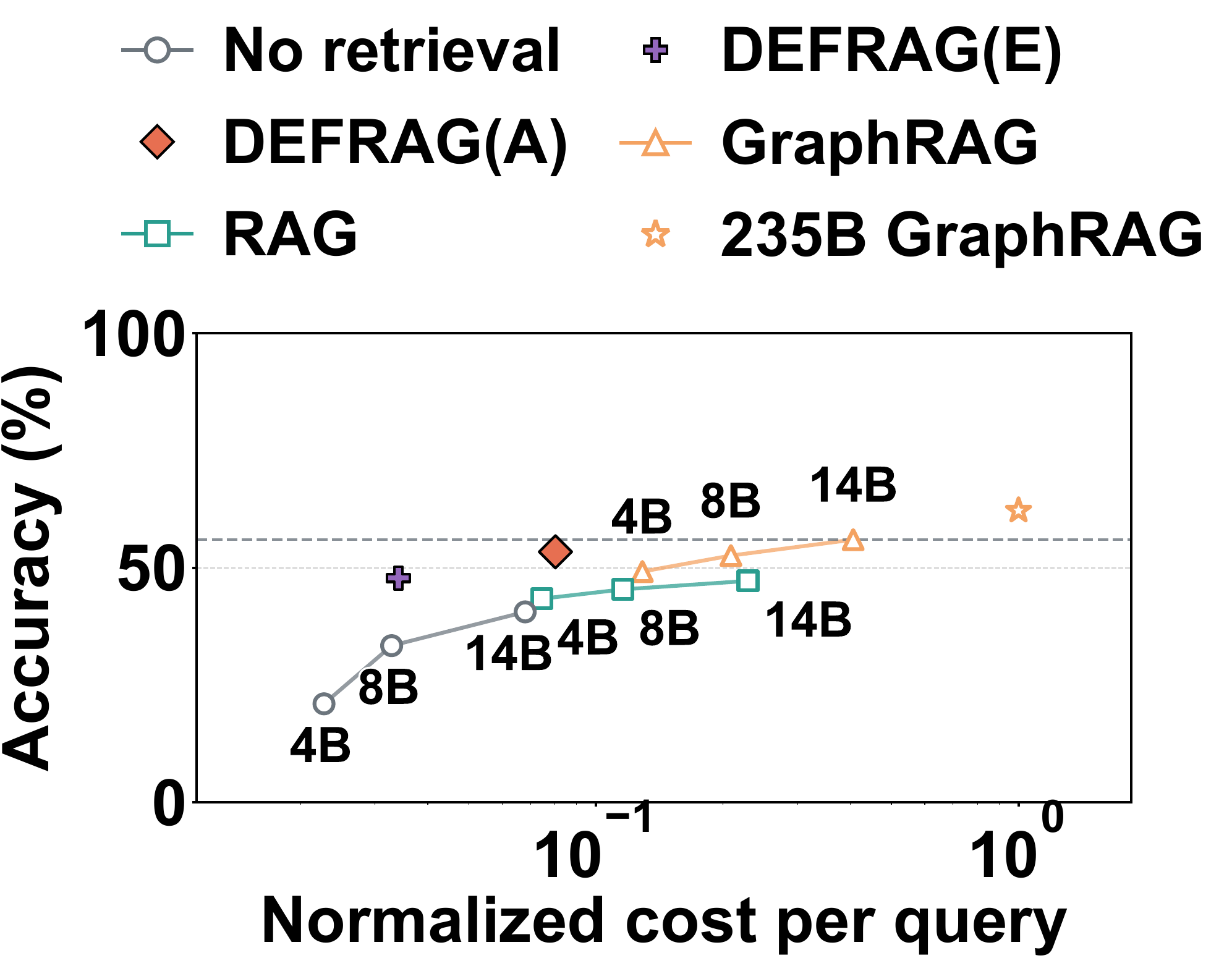}
    \caption{Accuracy--cost tradeoff.}
    \label{fig:revision_hp_generalization_frontier}
\end{subfigure}
\hfill
\begin{subfigure}[t]{0.377\linewidth}
    \centering
    \includegraphics[width=\linewidth]{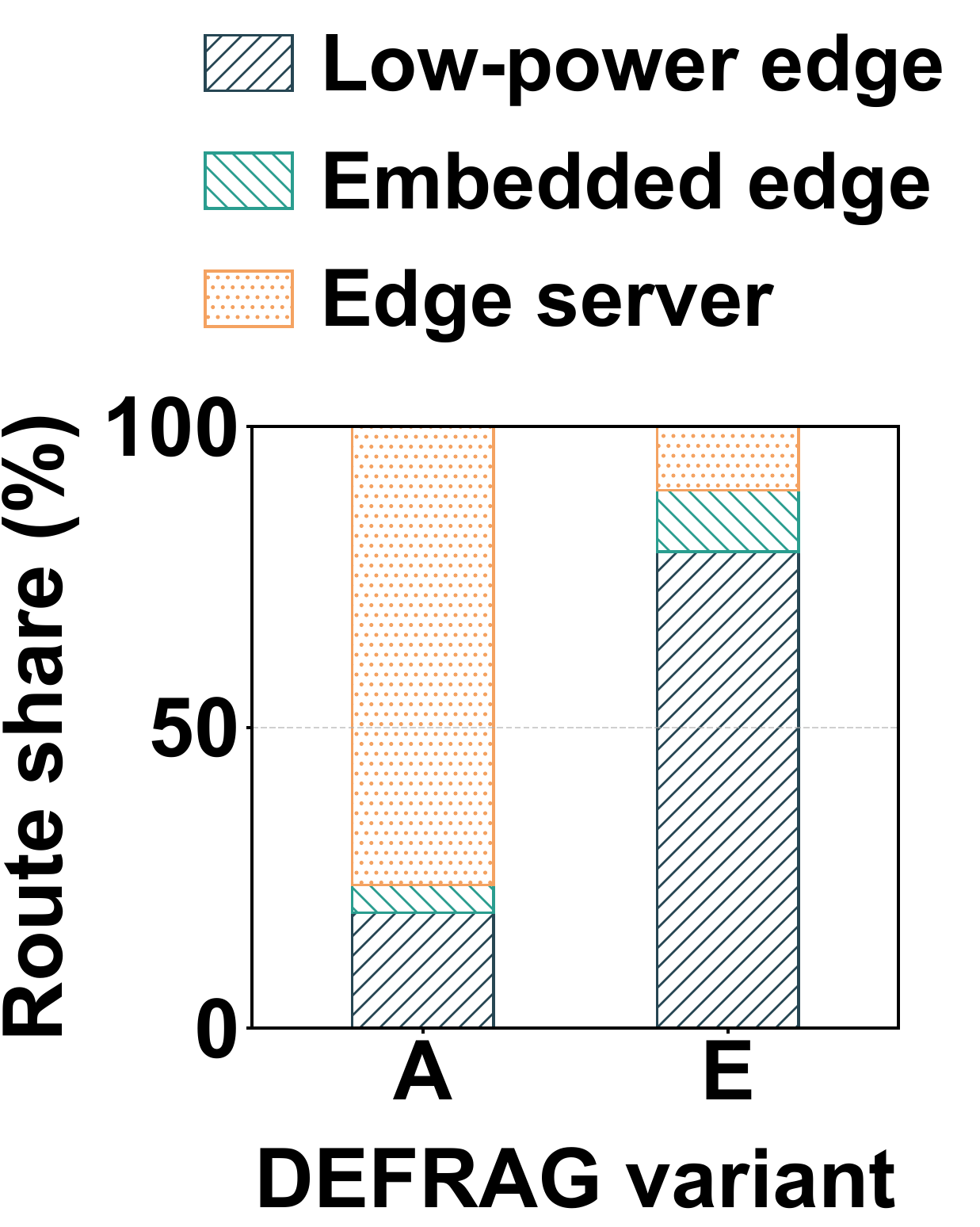}
    \caption{Route distribution.}
    \label{fig:revision_hp_generalization_routes}
\end{subfigure}
\caption{HarryPotter QA accuracy--cost tradeoff and route distribution.}
\label{fig:revision_hp_generalization}
\end{figure}

The three placements give comparable answer accuracy, indicating that collaborative retrieval still finds usable evidence when data placement is non-uniform. We then measure the communication overhead of this retrieval process from the requester origin. Capacity-aware placement is a capacity-oriented skewed baseline. It assigns frequently requested evidence groups to higher-capacity peers, yielding a 90.5\% requester-origin remote-fetch rate and 2.24\,s mean retrieval latency. Requester-aware placement uses requester demand as a placement signal. It places evidence closer to the requester origins that most often need it, reducing remote fetch to 61.3\% and mean retrieval latency to 0.85\,s, compared with 76.6\% and 2.03\,s for Balanced. Together, these results show that DEFRAG preserves answer quality without a perfectly balanced partition. The requester-aware result also demonstrates the potential of using requester demand to reduce communication overhead under skewed storage.

Finally, we add HarryPotter QA as a domain-specific generalization task. The task uses a local seven-book narrative corpus with reference QA pairs. It adds domain-specific evidence and narrative entities to the benchmark suite. Reference answers allow a controlled accuracy--cost comparison.

Table~\ref{tab:revision_hp_comparison} and Fig.~\ref{fig:revision_hp_generalization}(a) compare HarryPotter QA operating points across no retrieval, fixed RAG, fixed GraphRAG, DEFRAG(A/E), and cloud GraphRAG. The strongest local baseline is 14B GraphRAG, which reaches 56.0\% accuracy at normalized cost 0.4062. DEFRAG(A) reaches 53.4\% accuracy at cost 0.0802, or about 95\% of the local 14B GraphRAG accuracy at about one fifth of its cost. It also reaches 86\% of the cloud GraphRAG accuracy at about 8\% of the cloud cost. DEFRAG(E) moves further toward efficiency: its normalized cost is 0.0341, about 15\% of fixed RAG cost and 3\% of cloud GraphRAG cost.

Figure~\ref{fig:revision_hp_generalization}(b) explains this tradeoff through the selected model tiers. DEFRAG(A) routes 76\% of queries to the 14B edge-server tier, while DEFRAG(E) routes 79\% to the 1.7B low-power tier. This mirrors the NQ and HotpotQA results: the accuracy objective uses stronger models more often, while the efficiency objective keeps a lower-cost operating point. The HarryPotter result therefore supports DEFRAG's accuracy--cost tradeoff on an additional domain-specific QA corpus.

\section{Related Work}
\textbf{Cost-Effective Large Language Model Usage.} Reducing LLM deployment costs is an active research focus. Techniques like model quantization~\cite{xiao2023smoothquant,park2023lut} and pruning~\cite{ma2024llm}, distillation~\cite{gu2024minillm} cut costs but may reduce performance and require specialized hardware. Other strategies include caching LLM responses for frequent queries~\cite{zhu2023optimal,li2024scalm} or reusing key-value states during inference~\cite{liu2024cachegen,yao2024cacheblend}. Additionally, model multiplexing dynamically selects an appropriate model based on real-time query analysis~\cite{kim2023speculative,chenfrugalgpt,jeong2024adaptive}. However, these approaches primarily focus on reducing resource consumption during inference. They do not adequately address deploying LLMs effectively on edge devices with idle resources, particularly in achieving performance comparable to cloud-scale models.

\textbf{Retrieval-Augmented Generation (RAG) Enhancements.} RAG integrates external textual knowledge to enhance language model outputs~\cite{gao2023retrieval,singh2025agentic}. Early implementations, termed \textit{Naive RAG}, primarily employ sparse retrieval methods like TF-IDF and BM25, directly concatenating passages into prompts~\cite{lewis2020retrieval}. Later developments, or \textit{Advanced RAG}, incorporate dense retrieval such as Dense Passage Retrieval (DPR)~\cite{karpukhin2020dense} and iterative refinement introduce iterative feedback and complexity-based query refinement, improving multi-hop reasoning~\cite{yoran2023making,asai2024self,yan2024corrective,jeong2024adaptive}. \textit{Modular RAG} frameworks emphasize interchangeable modules for retrieval, reranking, and generation, facilitating hybrid strategies and external integrations~\cite{jin2025flashrag}. Furthermore, \textit{Graph-based RAG} methods utilize knowledge graph structures to enhance relational reasoning and structured retrieval for complex tasks~\cite{edge2024local,peng2024graph,liang2025kag}. While RAG methods can enhance accuracy, they often do so at the cost of increased resource consumption. Identifying effective trade-offs and selecting suitable RAG strategies for efficient deployment on edge nodes remains a key challenge. 

\textbf{LLM Resource Allocation Strategies in Edge Computing.}
Research on deploying LLMs in edge computing emphasizes optimal resource allocation by balancing latency, energy efficiency, and computational capabilities through techniques like task offloading, resource scheduling, and cloud-edge collaboration~\cite{wang2024end,qu2025mobile}. Multi-stage inference, for example, initially employs smaller models deployable on edge devices to handle general queries, escalating to larger LLMs when confidence is insufficient~\cite{wang2023tabi}. Active inference strategies dynamically offload resource-intensive LLM inference tasks across multiple edge nodes based on network bandwidth and processing capabilities~\cite{he2024large,zhang2024edgeshard}. Distributed on-device inference systems such as PRIMA.CPP demonstrate fast 30B--70B LLM inference on heterogeneous low-resource home clusters~\cite{li2026primacpp}. Incorporating RAG techniques at the edge further enhances LLM performance through federated multi-domain knowledge bases, adaptive retrieval and inference strategies, and by explicitly integrating logical relationships derived from knowledge graphs into context~\cite{shojaee2025federated,li2024eaco,ouyang2025adarag,fan2025minirag}. Edge computing with RAG has also been explored in emerging domains such as next-generation networks~\cite{huang2024toward}. Nevertheless, due to memory constraints on edge devices, knowledge coverage remains significantly limited. Therefore, efficient collaboration among decentralized edge nodes remains an open challenge and is crucial for achieving cloud-scale LLM performance.

Prior work typically optimizes retrieval or generation separately at the edge. By contrast, DEFRAG jointly optimizes both components, explicitly modeling the interaction between retrieval strategies and generation model capabilities. This unified system architecture enables more effective adaptation to resource constrained edge devices.

\section{Future Work}
This work takes an important step toward scalable and accessible LLM services at the edge. Looking ahead, several promising directions merit further exploration. Advancing dynamic integration and grouping mechanisms will be critical for supporting evolving edge networks, including automatic onboarding of new devices and efficient handling of cold-start scenarios. As knowledge bases and network scale increase, optimizing local caching and update strategies for lightweight knowledge graphs and vector stores will be essential; adaptive cache management that responds to storage constraints and access patterns will be a valuable avenue. In addition, advancing inference efficiency and seamless system integration on mobile devices will remain crucial to enable practical applications facing users.

\section{Conclusion}
In this paper, we have introduced DEFRAG, a decentralized edge collaboration architecture for RAG that enables accurate, cost-efficient, and scalable LLM deployment at the edge. By compressing and distributing lightweight knowledge graphs for retrieval and dynamically tuning generation parameters with Bayesian online learning, DEFRAG narrows the accuracy gap between edge-based SLMs and cloud-based LLMs while reducing costs and improving scalability. By enabling collaborative sharing of models and knowledge among the crowd, DEFRAG advances the democratization of LLM services from cloud to crowd.

\bibliographystyle{IEEEtran}
\bibliography{reference}

@inproceedings{xiao2023smoothquant,
      title={Smoothquant: Accurate and efficient post-training quantization for large language models},
      author={Xiao, Guangxuan and Lin, Ji and Seznec, Mickael and Wu, Hao and Demouth, Julien and Han, Song},
      booktitle={International Conference on Machine Learning (ICML)},
      year={2023},
    }

@inproceedings{park2023lut,
      title={Lut-gemm: Quantized matrix multiplication based on luts for efficient inference in large-scale generative language models},
      author={Park, Gunho and Kim, Minsub and Lee, Sungjae and Kim, Jeonghoon and Kwon, Beomseok and Kwon, Se Jung and Kim, Byeongwook and Lee, Youngjoo and Lee, Dongsoo and others},
      booktitle={International Conference on Learning Representations (ICLR)},
      year={2023}
    }

@article{ma2024llm,
      title={Llm-pruner: On the structural pruning of large language models},
      author={Ma, Xinyin and Fang, Gongfan and Wang, Xinchao},
      journal={International Conference on Neural Information Processing Systems (NeurIPS)},
      year={2024}
    }

@inproceedings{gu2024minillm,
      title={MiniLLM: Knowledge Distillation of Large Language Models},
      author={Gu, Yuxian and Dong, Li and Wei, Furu and Huang, Minlie},
      booktitle={International Conference on Learning Representations (ICLR)},
      year={2024}
    }

@inproceedings{zhu2023optimal,
      title={On optimal caching and model multiplexing for large model inference},
      author={Zhu, Banghua and Sheng, Ying and Zheng, Lianmin and Barrett, Clark and Jordan, Michael I and Jiao, Jiantao},
      booktitle={International Conference on Neural Information Processing Systems (NeurIPS)},
      year={2023}
    }

@inproceedings{li2024scalm,
      title={Scalm: Towards semantic caching for automated chat services with large language models},
      author={Li, Jiaxing and Xu, Chi and Wang, Feng and von Riedemann, Isaac M and Zhang, Cong and Liu, Jiangchuan},
      booktitle={International Symposium on Quality of Service (IWQoS)},
      year={2024},
    }

@inproceedings{liu2024cachegen,
      title={CacheGen: KV Cache Compression and Streaming for Fast Large Language Model Serving},
      author={Liu, Yuhan and Li, Hanchen and Cheng, Yihua and Ray, Siddhant and Huang, Yuyang and Zhang, Qizheng and Du, Kuntai and Yao, Jiayi and Lu, Shan and Ananthanarayanan, Ganesh and others},
      booktitle={International Conference on Applications,  Technologies, Architectures, and Protocols for  Computer Communication (SIGCOMM)},
      year={2024}
    }

@article{yao2024cacheblend,
      title={CacheBlend: Fast Large Language Model Serving with Cached Knowledge Fusion},
      author={Yao, Jiayi and Li, Hanchen and Liu, Yuhan and Ray, Siddhant and Cheng, Yihua and Zhang, Qizheng and Du, Kuntai and Lu, Shan and Jiang, Junchen},
      journal={arXiv preprint arXiv:2405.16444},
      year={2024}
    }

@article{kim2023speculative,
      title={Speculative decoding with big little decoder},
      author={Kim, Sehoon and Mangalam, Karttikeya and Moon, Suhong and Malik, Jitendra and Mahoney, Michael W and Gholami, Amir and Keutzer, Kurt},
      journal={International Conference on Neural Information Processing Systems (NeurIPS)},
      year={2023}
    }

@inproceedings{jeong2024adaptive,
      title={Adaptive-RAG: Learning to Adapt Retrieval-Augmented Large Language Models through Question Complexity},
      author={Jeong, Soyeong and Baek, Jinheon and Cho, Sukmin and Hwang, Sung Ju and Park, Jong-Cheol},
      booktitle={Conference of the North American Chapter of the Association for Computational Linguistics (NAACL)},
      year={2024},
    }

@article{chenfrugalgpt,
      title={FrugalGPT: How to Use Large Language Models While Reducing Cost and Improving Performance},
      author={Chen, Lingjiao and Zaharia, Matei and Zou, James},
      journal={Transactions on Machine Learning Research (TMLR)},
      year={2025},
    }

@article{gao2023retrieval,
      title={Retrieval-augmented generation for large language models: A survey},
      author={Gao, Yunfan and Xiong, Yun and Gao, Xinyu and Jia, Kangxiang and Pan, Jinliu and Bi, Yuxi and Dai, Yixin and Sun, Jiawei and Wang, Haofen and Wang, Haofen},
      journal={arXiv preprint arXiv:2312.10997},
      year={2023}
    }

@article{singh2025agentic,
      title={Agentic Retrieval-Augmented Generation: A Survey on Agentic RAG},
      author={Singh, Aditi and Ehtesham, Abul and Kumar, Saket and Khoei, Tala Talaei},
      journal={arXiv preprint arXiv:2501.09136},
      year={2025}
    }

@article{lewis2020retrieval,
      title={Retrieval-augmented generation for knowledge-intensive nlp tasks},
      author={Lewis, Patrick and Perez, Ethan and Piktus, Aleksandra and Petroni, Fabio and Karpukhin, Vladimir and Goyal, Naman and K{\"u}ttler, Heinrich and Lewis, Mike and Yih, Wen-tau and Rockt{\"a}schel, Tim and others},
      journal={International Conference on Neural Information Processing Systems (NeurIPS)},
      year={2020}
    }

@inproceedings{karpukhin2020dense,
      title={Dense Passage Retrieval for Open-Domain Question Answering},
      author={Karpukhin, Vladimir and Oguz, Barlas and Min, Sewon and Lewis, Patrick and Wu, Ledell and Edunov, Sergey and Chen, Danqi and Yih, Wen-tau},
      booktitle={Conference on Empirical Methods in Natural Language Processing (EMNLP)},
      year={2020},
    }

@inproceedings{asai2024self,
      title={Self-RAG: Learning to Retrieve, Generate, and Critique through Self-Reflection},
      author={Asai, Akari and Wu, Zeqiu and Wang, Yizhong and Sil, Avi and Hajishirzi, Hannaneh},
      booktitle={International Conference on Learning Representations (ICLR)},
      year={2024}
    }

@article{yan2024corrective,
      title={Corrective retrieval augmented generation},
      author={Yan, Shi-Qi and Gu, Jia-Chen and Zhu, Yun and Ling, Zhen-Hua},
      journal={arXiv preprint arXiv:2401.15884},
      year={2024}
    }

@inproceedings{yoran2023making,
      title={Making retrieval-augmented language models robust to irrelevant context},
      author={Yoran, Ori and Wolfson, Tomer and Ram, Ori and Berant, Jonathan},
      booktitle={International Conference on Learning Representations (ICLR)},
      year={2024}
    }

@inproceedings{jin2025flashrag,
      title={Flashrag: A modular toolkit for efficient retrieval-augmented generation research},
      author={Jin, Jiajie and Zhu, Yutao and Dou, Zhicheng and Dong, Guanting and Yang, Xinyu and Zhang, Chenghao and Zhao, Tong and Yang, Zhao and Wen, Ji-Rong},
      booktitle={Companion Proceedings of the ACM on Web Conference},
      year={2025}
    }

@article{edge2024local,
      title={From local to global: A graph rag approach to query-focused summarization},
      author={Edge, Darren and Trinh, Ha and Cheng, Newman and Bradley, Joshua and Chao, Alex and Mody, Apurva and Truitt, Steven and Metropolitansky, Dasha and Ness, Robert Osazuwa and Larson, Jonathan},
      journal={arXiv preprint arXiv:2404.16130},
      year={2024}
    }

@article{peng2024graph,
      title={Graph retrieval-augmented generation: A survey},
      author={Peng, Boci and Zhu, Yun and Liu, Yongchao and Bo, Xiaohe and Shi, Haizhou and Hong, Chuntao and Zhang, Yan and Tang, Siliang},
      journal={arXiv preprint arXiv:2408.08921},
      year={2024}
    }

@inproceedings{liang2025kag,
      title={Kag: Boosting llms in professional domains via knowledge augmented generation},
      author={Liang, Lei and Bo, Zhongpu and Gui, Zhengke and Zhu, Zhongshu and Zhong, Ling and Zhao, Peilong and Sun, Mengshu and Zhang, Zhiqiang and Zhou, Jun and Chen, Wenguang and others},
      booktitle={Companion Proceedings of the ACM on Web Conference},
      year={2025}
    }

@article{wang2024end,
      title={End-edge-cloud collaborative computing for deep learning: A comprehensive survey},
      author={Wang, Yingchao and Yang, Chen and Lan, Shulin and Zhu, Liehuang and Zhang, Yan},
      journal={IEEE Communications Surveys \& Tutorials (ICST)},
      year={2024},
    }

@article{qu2025mobile,
      title={Mobile edge intelligence for large language models: A contemporary survey},
      author={Qu, Guanqiao and Chen, Qiyuan and Wei, Wei and Lin, Zheng and Chen, Xianhao and Huang, Kaibin},
      journal={IEEE Communications Surveys \& Tutorials (ICST)},
      year={2025},
      publisher={IEEE}
    }

@inproceedings{wang2023tabi,
      title={Tabi: An efficient multi-level inference system for large language models},
      author={Wang, Yiding and Chen, Kai and Tan, Haisheng and Guo, Kun},
      booktitle={European Conference on Computer Systems (EuroSys)},
      year={2023}
    }

@article{he2024large,
      title={Large language models (LLMs) inference offloading and resource allocation in cloud-edge computing: An active inference approach},
      author={He, Ying and Fang, Jingcheng and Yu, F Richard and Leung, Victor C},
      journal={IEEE Transactions on Mobile Computing (TMC)},
      year={2024},
    }

@article{zhang2024edgeshard,
      title={Edgeshard: Efficient llm inference via collaborative edge computing},
      author={Zhang, Mingjin and Shen, Xiaoming and Cao, Jiannong and Cui, Zeyang and Jiang, Shan},
      journal={IEEE Internet of Things Journal (IOT)},
      year={2024},
    }

@inproceedings{shojaee2025federated,
      title={Federated Retrieval Augmented Generation for Multi-Product Question Answering},
      author={Shojaee, Parshin and Harsha, Sai Sree and Luo, Dan and Maharaj, Akash and Yu, Tong and Li, Yunyao},
      booktitle={International Conference on Computational Linguistics (COLING)},
      year={2025}
    }

@inproceedings{ouyang2025adarag,
      title={AdaRAG: Adaptive Optimization for Retrieval Augmented Generation with Multilevel Retrievers at the Edge},
      author={Tao Ouyang and Guihang Hong and Kongyange Zhao and Zhi Zhou and Weigang Wu and Zhaobiao Lv and Xu Chen},
      booktitle = {IEEE International Conference on Computer Communications (INFOCOM)},
      year= {2025},
    }

@article{li2024eaco,
      title={EACO-RAG: Edge-Assisted and Collaborative RAG with Adaptive Knowledge Update},
      author={Li, Jiaxing and Xu, Chi and Jia, Lianchen and Wang, Feng and Zhang, Cong and Liu, Jiangchuan},
      journal={arXiv preprint arXiv:2410.20299},
      year={2024}
    }

@article{huang2024toward,
      title={Toward Effective Retrieval Augmented Generative Services in 6G Networks},
      author={Huang, Xi and Tang, Yinxu and Li, Junling and Zhang, Ning and Shen, Xuemin Sherman},
      journal={IEEE Network},
      year={2024},
    }

@article{fan2025minirag,
      title={MiniRAG: Towards Extremely Simple Retrieval-Augmented Generation},
      author={Fan, Tianyu and Wang, Jingyuan and Ren, Xubin and Huang, Chao},
      journal={arXiv preprint arXiv:2501.06713},
      year={2025}
    }

@inproceedings{li2026primacpp,
      title={{\revise{Prima.cpp: Fast 30--70B LLM Inference on Heterogeneous and Low-Resource Home Clusters}}},
      author={{\revise{Z. Li, T. Li, W. Feng, R. Xiao, J. She, H. Huang, M. Guizani, H. Yu, Q. Ho, W. Xiang, and X. Liu}}},
      booktitle={{\revise{International Conference on Learning Representations (ICLR)}}},
      year={{\revise{2026}}}
    }

@inproceedings{yang2018hotpotqa,
      title={{HotpotQA}: A Dataset for Diverse, Explainable Multi-hop Question Answering},
      author={Yang, Zhilin and Qi, Peng and Zhang, Saizheng and Bengio, Yoshua and Cohen, William W. and Salakhutdinov, Ruslan and Manning, Christopher D.},
      booktitle={Conference on Empirical Methods in Natural Language Processing ({EMNLP})},
      year={2018}
    }

@article{kwiatkowski2019natural,
      title={Natural questions: a benchmark for question answering research},
      author={Kwiatkowski, Tom and Palomaki, Jennimaria and Redfield, Olivia and Collins, Michael and Parikh, Ankur and Alberti, Chris and Epstein, Danielle and Polosukhin, Illia and Devlin, Jacob and Lee, Kenton and others},
      journal={Transactions of the Association for Computational Linguistics (TACL)},
      year={2019},
    }

@inproceedings{jin2024long,
  title={Long-context llms meet rag: Overcoming challenges for long inputs in rag},
  author={Jin, Bowen and Yoon, Jinsung and Han, Jiawei and Arik, Sercan O},
  booktitle={International Conference on Learning Representations (ICLR)},
  year={2025},
}

@inproceedings{lan2010axiomatic,
      title={An axiomatic theory of fairness in network resource allocation},
      author={Lan, Tian and Kao, David and Chiang, Mung and Sabharwal, Ashutosh},
      booktitle = {IEEE International Conference on Computer Communications (INFOCOM)},
      year={2010},
    }

@inproceedings{wangpandalm,
  title={PandaLM: An Automatic Evaluation Benchmark for LLM Instruction Tuning Optimization},
  author={Wang, Yidong and Yu, Zhuohao and Yao, Wenjin and Zeng, Zhengran and Yang, Linyi and Wang, Cunxiang and Chen, Hao and Jiang, Chaoya and Xie, Rui and Wang, Jindong and others},
  booktitle={International Conference on Learning Representations (ICLR)},
  year={2024},
}

@misc{llama.cpp,
      author       = {Georgi Gerganov and contributors},
      title        = {llama.cpp: Port of Facebook's LLaMA model in C/C++},
      year         = {2023},
      howpublished = {\url{https://github.com/ggml-org/llama.cpp}},
    }

@article{qwen3,
        title={Qwen3 Technical Report}, 
        author={An Yang and Anfeng Li and Baosong Yang and Beichen Zhang and Binyuan Hui and Bo Zheng and Bowen Yu and Chang Gao and Chengen Huang and Chenxu Lv and others},
        journal = {arXiv preprint arXiv:2505.09388},
        year={2025}
    }

@misc{openai2025gpt41api,
      author       = {OpenAI},
      title        = {Introducing GPT-4.1 in the API},
      year         = {2025},
      howpublished = {\url{https://openai.com/index/gpt-4-1/}},
    }

@misc{flask-docs,
      author       = {Armin Ronacher and the Pallets team},
      title        = {Flask Documentation},
      organization = {Pallets Projects},
      year         = {2025},
      howpublished = {\url{https://flask.palletsprojects.com/}},
    }

@misc{ollama-api-docs,
      author       = {Ollama Team},
      title        = {Ollama API Documentation},
      year         = {2025},
      howpublished =  {\url{https://github.com/jmorganca/ollama/blob/main/docs/api.md}},
    }

@book{williams2006gaussian,
      title={Gaussian processes for machine learning},
      author={Williams, Christopher KI and Rasmussen, Carl Edward},
      publisher={MIT press Cambridge, MA},
      year={2006}
    }

@book{cover1999elements,
    author = {Cover, Thomas M. and Thomas, Joy A.},
    title = {Elements of information theory},
    year = {1991},
    publisher = {Wiley-Interscience},
    }

@inproceedings{srinivas2009gaussian,
      title={Gaussian process optimization in the bandit setting: No regret and experimental design},
      author={Srinivas, Niranjan and Krause, Andreas and Kakade, Sham M and Seeger, Matthias},
      booktitle={International Conference on Machine Learning (ICML)},
      year={2010}
    }

@inproceedings{li2010contextual,
      title={A contextual-bandit approach to personalized news article recommendation},
      author={Li, Lihong and Chu, Wei and Langford, John and Schapire, Robert E},
      booktitle={International World Wide Web Conferences (WWW)},
      year={2010}
    }

@article{cesa2013gang,
      title={A gang of bandits},
      author={Cesa-Bianchi, Nicolo and Gentile, Claudio and Zappella, Giovanni},
      journal={International Conference on Neural Information Processing Systems (NeurIPS)},
      year={2013}
    }

@article{bull2011convergence,
      title={Convergence rates of efficient global optimization algorithms.},
      author={Bull, Adam D},
      journal={Journal of Machine Learning Research (JMLR)},
      year={2011}
    }

@inproceedings{schaul2015prioritized,
      title={Prioritized experience replay},
      author={Schaul, Tom and Quan, John and Antonoglou, Ioannis and Silver, David},
      booktitle={International Conference on Learning Representations (ICLR)},
      year={2016}
    }

@inproceedings{zhu2024accelerating,
  title={Accelerating inference of retrieval-augmented generation via sparse context selection},
  author={Zhu, Yun and Gu, Jia-Chen and Sikora, Caitlin and Ko, Ho and Liu, Yinxiao and Lin, Chu-Cheng and Shu, Lei and Luo, Liangchen and Meng, Lei and Liu, Bang and others},
  booktitle={International Conference on Learning Representations (ICLR)},
  year={2025}
}

@inproceedings{groeneveld2020simple,
  title={A Simple Yet Strong Pipeline for HotpotQA},
  author={Groeneveld, Dirk and Khot, Tushar and Sabharwal, Ashish and others},
  booktitle={Conference on Empirical Methods in Natural Language Processing (EMNLP)},
  year={2020}
}

@inproceedings{samsi2023words,
      title={From words to watts: Benchmarking the energy costs of large language model inference},
      author={Samsi, Siddharth and Zhao, Dan and McDonald, Joseph and Li, Baolin and Michaleas, Adam and Jones, Michael and Bergeron, William and Kepner, Jeremy and Tiwari, Devesh and Gadepally, Vijay},
      booktitle={IEEE High Performance Extreme Computing Conference (HPEC)},
      year={2023}
    }

@inproceedings{reidys2025coach,
      title={Coach: Exploiting temporal patterns for all-resource oversubscription in cloud platforms},
      author={Reidys, Benjamin and Zardoshti, Pantea and Goiri, {\'I}{\~n}igo and Irvene, Celine and Berger, Daniel S and Ma, Haoran and Arya, Kapil and Cortez, Eli and Stark, Taylor and Bak, Eugene and others},
      booktitle={International Conference on Architectural Support for Programming Languages and Operating Systems (ASPLOS)},
      year={2025}
    }

@inproceedings{zhao2023scalable,
      title={Scalable tail latency estimation for data center networks},
      author={Zhao, Kevin and Goyal, Prateesh and Alizadeh, Mohammad and Anderson, Thomas E},
      booktitle={USENIX Symposium on Networked Systems Design and Implementation (NSDI)},
      year={2023}
    }

\begin{IEEEbiography}
    [{\includegraphics[width=1in,height=1.25in,clip,keepaspectratio]{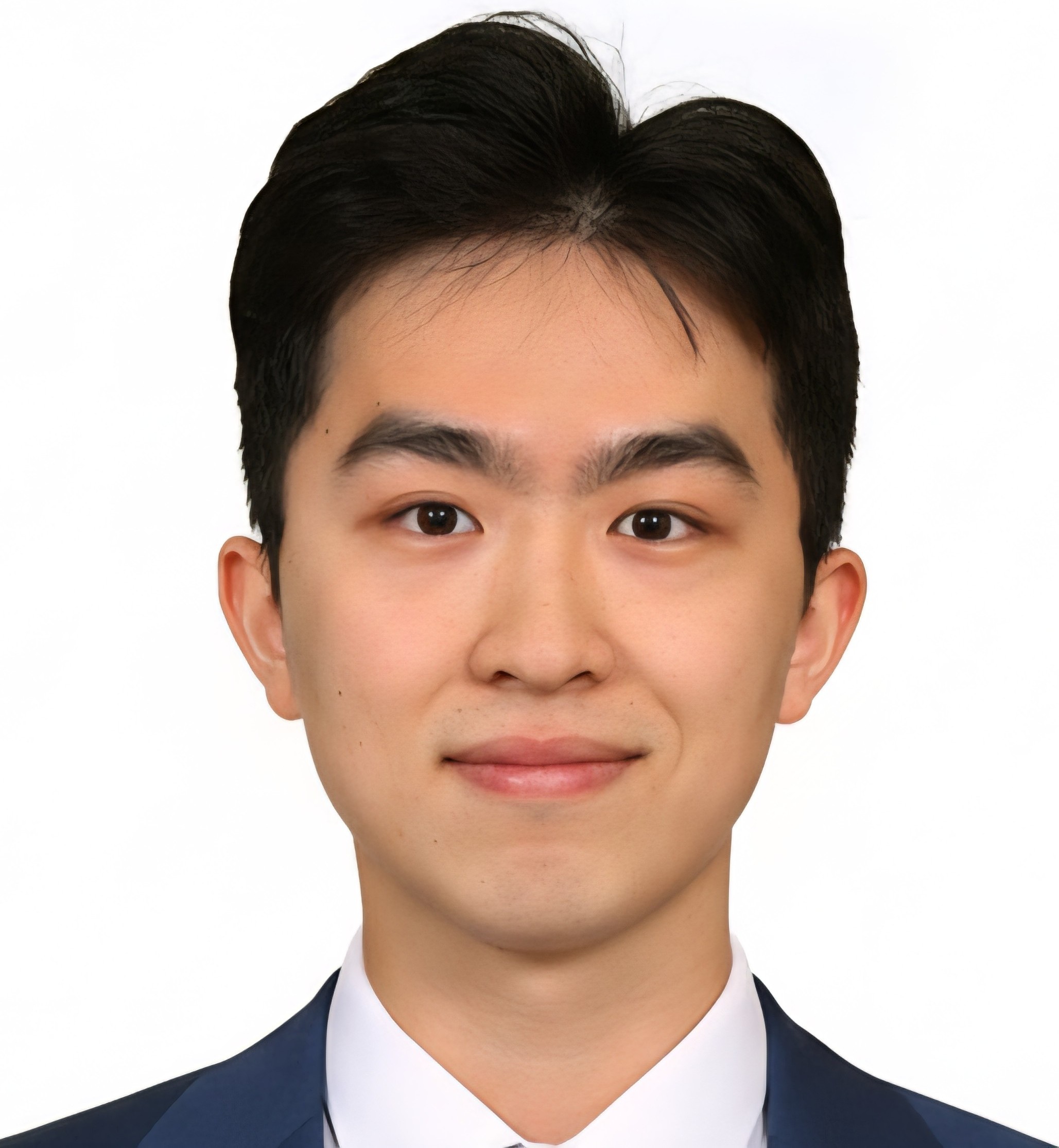}}]{Jiaxing Li}
    is currently a Ph.D. candidate in the School of Computing Science at Simon Fraser University, British Columbia, Canada, under the supervision of Prof. Jiangchuan Liu. He received his B.Sc. degree in Computing Science from Simon Fraser University in 2023. His research interests include large language model applications, cloud edge computing, and distributed machine learning systems.
\end{IEEEbiography}

\begin{IEEEbiography}
    [{\includegraphics[width=1in,height=1.25in,clip,keepaspectratio]{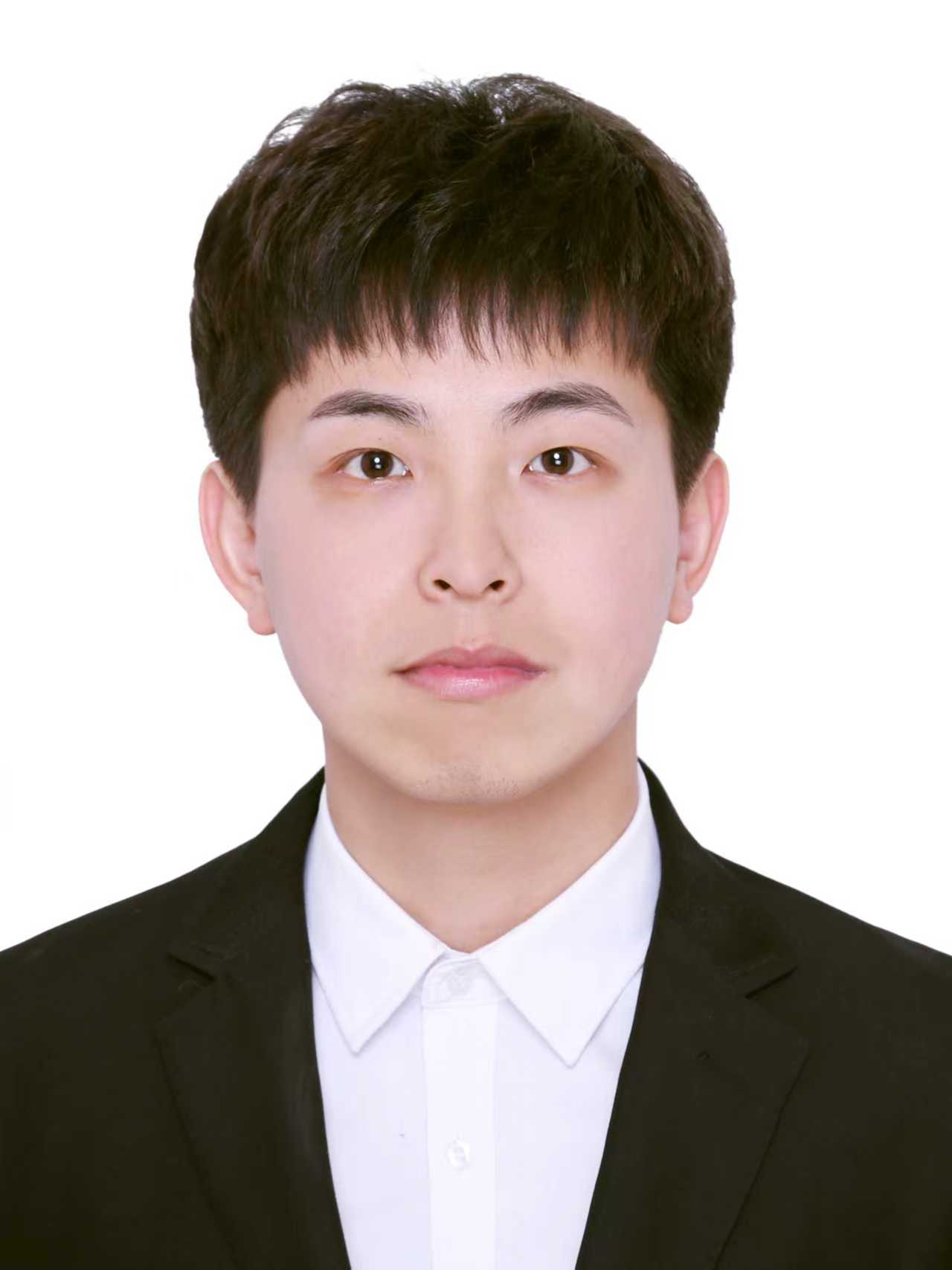}}]{Hengzhi Wang} 
    (Member, IEEE) is currently serving as an assistant professor in the College of Computer Science and Software Engineering at Shenzhen University in Shenzhen, China. He obtained his B.S. degree in software engineering from Jilin University in Changchun, China, in 2017, and his Ph.D. degree in computer science from Jilin University in Changchun, China, in 2023. During his doctoral studies, he also worked as a visiting Ph.D. student at the School of Computing Science, Simon Fraser University in British Columbia, Canada. His research interests primarily focus on spatial crowdsourcing, federated learning, and privacy protection.
\end{IEEEbiography}

\begin{IEEEbiography}
    [{\includegraphics[width=1in,height=1.25in,clip,keepaspectratio]{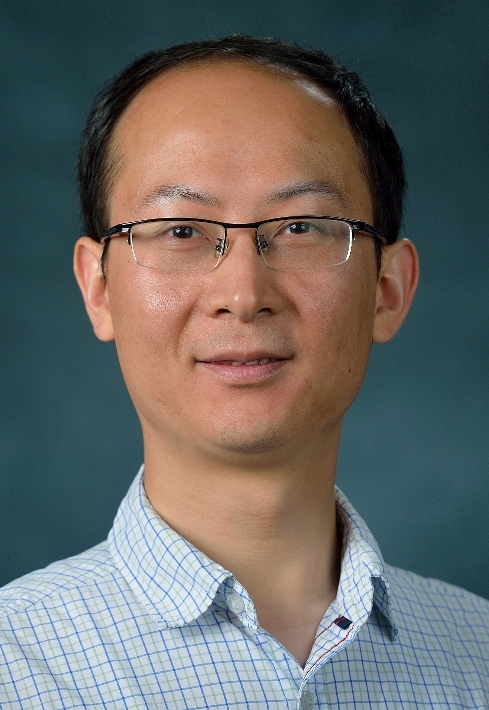}}]{Feng Wang}
    (Senior Member, IEEE) received both the Bachelor's degree and Master's degree in Computer Science and Technology from Tsinghua University, Beijing, China in 2002 and 2005, respectively. He received the PhD degree in Computing Science from Simon Fraser University, Burnaby, British Columbia, Canada in 2012. He is currently an Associate Professor in the Department of Computer and Information Science at the University of Mississippi, University, MS, USA. He is a Senior Member of IEEE. He is a recipient of IEEE ICME Quality Reviewer Award (2011), ACM BuildSys Best Paper Award (2018), ACM/IEEE IWQoS Best Paper Award Runner-up Award (2022), IEEE CloudCom Special Section Best Paper Award (2025), and IEEE INFOCOM Distinguished TPC Member Awards (2020-2022). He is a Technical Committee Member of Elsevier Computer Communications. He served as Program Vice Chair in International Conference on Internet of Vehicles (IOV) 2014, as TPC Co-chair in IEEE CloudCom 2017 for Internet of Things and Mobile on Cloud track, as Publication Co-chair in ACM/IEEE IWQoS 2024. He also serves as TPC member in various international conferences such as IEEE INFOCOM, IEEE/ACM IWQoS, ACM Multimedia, IEEE ICC, IEEE GLOBECOM and IEEE ICME.
\end{IEEEbiography}

\begin{IEEEbiography}
    [{\includegraphics[width=1in,height=1.25in,clip,keepaspectratio]{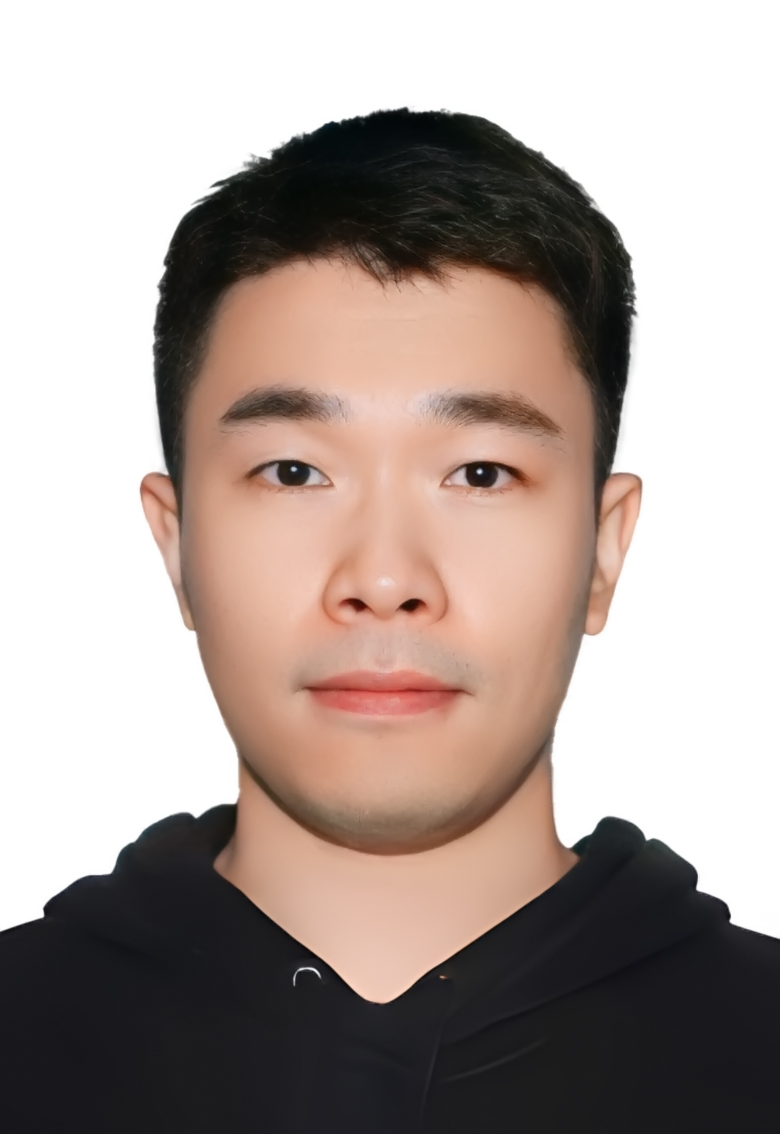}}]{Chi Xu}
    received the B.Sc. degree in software engineering from Xidian University, Xi'an, China, and the M.Sc. degree in computing science from Simon Fraser University. He is currently pursuing his Ph.D. degree in computing science at Simon Fraser University, Canada. His research spans applied artificial intelligence and networked systems, with a recent focus on Physical AI for challenging field environments.
\end{IEEEbiography}

\begin{IEEEbiography}
    [{\includegraphics[width=1in,height=1.25in,clip,keepaspectratio]{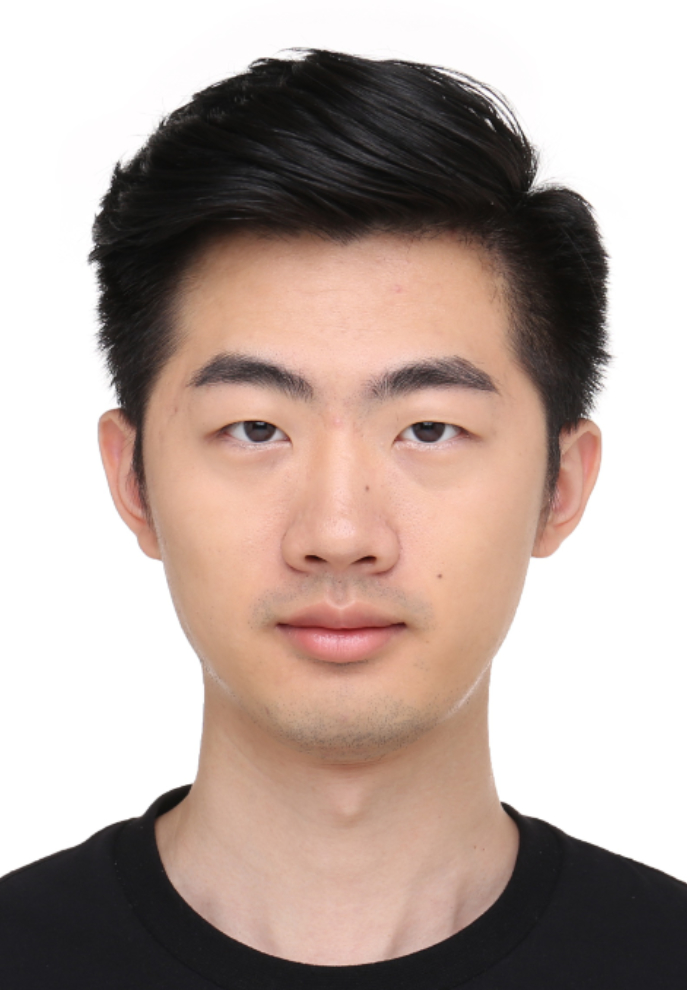}}]{Danyang Song}
    is currently a Ph.D. student in the Department of Electrical and Electronic Engineering at the University of Hong Kong. Prior to HKU, he received his MSc degree in Computing Science from Simon Fraser University in 2023 and Dual B.E degree in Computer Science from Zhejiang University and Simon Fraser University in 2018. His research interests include Artificial Intelligence, and edge computing under industrial scenarios.
\end{IEEEbiography}

\begin{IEEEbiography}
    [{\includegraphics[width=1in,height=1.25in,clip,keepaspectratio]{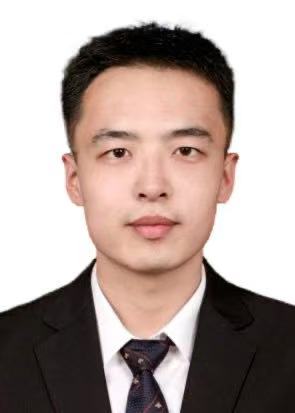}}]{Ruixiao Zhang}
    now is a Research Scientist at ByteDance Multimedia Lab in San Diego. Previously, he was a postdoctoral researcher at the University of Illinois Urbana-Champaign (supervised by Prof. Klara Nahrstedt), the University of Hong Kong (supervised by Prof. Chuan Wu), and the Hong Kong University of Science and Technology (supervised by Prof. Bo Li). He also collaborates closely with Prof. Jiangchuan Liu at Simon Fraser University. He received his Ph.D. in Computer Science from Tsinghua University in 2022, advised by Prof. Lifeng Sun, and his B.E. in Electronic Engineering from Tsinghua University in 2017. His research focuses on networking-driven optimization of multimedia systems, with an emphasis on robust and efficient video delivery over real-world networks under packet loss, bandwidth constraints, and time-varying conditions.
\end{IEEEbiography}

\begin{IEEEbiography}
    [{\includegraphics[width=1in,height=1.25in,clip,keepaspectratio]{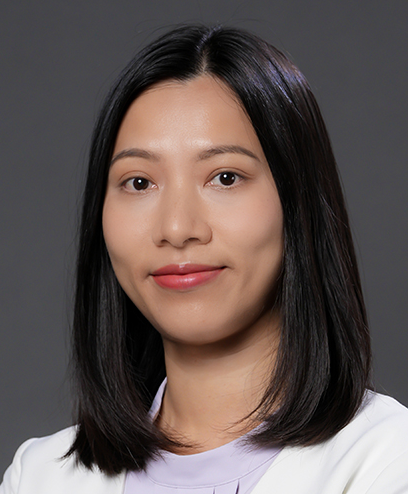}}]{Edith C.H. Ngai}
    (Senior Member, IEEE) is currently an Associate Professor in the Department of Electrical and Electronic Engineering, The University of Hong Kong. Before joining HKU in 2020, she was an Associate Professor in the Department of Information Technology, Uppsala University, Sweden. Her research interests include Internet-of-Things, edge intelligence, and smart cities. She was a VINNMER Fellow awarded by Swedish Governmental Research Funding Agency VINNOVA in 2009. Her co-authored papers received a Best Paper Award in QShine 2023, Best Paper Runner-Up Awards in IEEE IWQoS 2010 and ACM/IEEE IPSN 2013, and Best paper candidate in ACM BuildSys 2024. She was an Area Editor of IEEE Internet of Things Journal from 2020 to 2022. She is currently an Associate Editor in IEEE Transactions of Mobile Computing, IEEE Network, and IEEE Transactions of Industrial Informatics. She served as a program chair in ACM womENcourage 2015 and a TPC co-chair in IEEE SmartCity 2015, IEEE GreenCom 2022, IEEE/ACM IWQoS 2024, IEEE CloudCom 2025. She received a Meta Policy Research Award in Asia Pacific in 2022. She was selected as one of the N²Women Stars in Computer Networking and Communications in 2022. She is a Distinguished Lecturer in IEEE Communication Society in 2023-2024. She is an ACM Distinguished Member in 2025.
\end{IEEEbiography}

\begin{IEEEbiography}
    [{\includegraphics[width=1in,height=1.25in,clip,keepaspectratio]{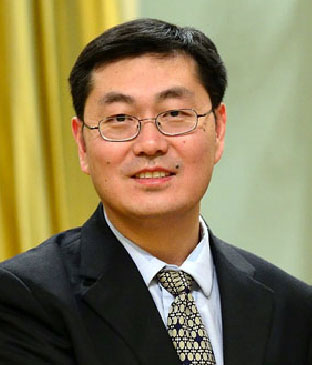}}]{Jiangchuan Liu}
    (Fellow, IEEE) is a Professor in the School of Computing Science, Simon Fraser University, British Columbia, Canada. He is a Fellow of Royal Society of Canada, a Fellow of The Canadian Academy of Engineering, an IEEE Fellow, and an NSERC E.W.R. Steacie Memorial Fellow. He was an EMC-Endowed Visiting Chair Professor of Tsinghua University (2013-2016) and is a Distinguished Guest Professor of Tsinghua Shenzhen International Graduate School (2022-). In the past he worked as an Assistant Professor at The Chinese University of Hong Kong (2003-2004) and as a research fellow at Microsoft Research Asia. He received the BEng degree (Magna cum laude) from Tsinghua University in 1999, and the PhD degree from The Hong Kong University of Science and Technology in 2003, both in computer science. He is a co-recipient of the inaugural Test of Time Paper Award of IEEE INFOCOM (2015), IEEE ICDCS Distinguished Paper Award (2024), ACM SIGMM TOMCCAP Nicolas D. Georganas Best Paper Award (2013), and ACM Multimedia Best Paper Award (2012). His research interests include intelligent multimedia computing and networking, cloud and edge computing, social networking, online gaming, and wireless mobile and space networking. He has served on the editorial boards of IEEE/ACM TON, IEEE TNSE, TMM, TBD, COMST, and IOTJ. He was a Steering Committee member of IEEE TMC, and Steering Committee Chair of IEEE/ACM IWQoS (2015-2017). He was TPC Chair of IEEE INFOCOM 2021 and General Chair of INFOCOM 2024. 
    
\end{IEEEbiography}

\end{document}